\documentclass[cmp]{svjour}
\pdfoutput=1

\usepackage{amsmath}
\usepackage{mathtools}
\usepackage{lmodern}
\usepackage[T1]{fontenc}
\usepackage{amsfonts}
\usepackage{bbm}
\DeclareMathAlphabet{\mathbfi}{OML}{cmm}{b}{it}

\let\originalleft\left
\let\originalright\right
\renewcommand{\left}{\mathopen{}\mathclose\bgroup\originalleft}
\renewcommand{\right}{\aftergroup\egroup\originalright}
\makeatletter
\newcommand{\vast}{\bBigg@{4}}
\newcommand{\Vast}{\bBigg@{5}}
\makeatother

\makeatletter
\newenvironment{equations}[1][]{\subequations\ifx\relax#1\relax\else\label{#1}\fi\align\ignorespaces}{\endalign\ignorespacesafterend\endsubequations}
\def\@spliteq#1{\begin{equation}\begin{split}#1\end{split}\end{equation}}
\def\@spliteqstar#1{\begin{equation*}\begin{split}#1\end{split}\end{equation*}}
\def\splitequation{\collect@body\@spliteq}
\expandafter\def\csname splitequation*\endcsname{\collect@body\@spliteqstar}

\expandafter\def\csname endsplitequation*\endcsname{\ignorespacesafterend}
\makeatother

\makeatletter
\def\normalthmheadings{\def\@spbegintheorem##1##2##3##4{\trivlist
                 \item[\hskip\labelsep{##3##1\ ##2\@thmcounterend}]##4\addcontentsline{toc}{subsection}{##1\ ##2\@thmcounterend}}
\def\@spopargbegintheorem##1##2##3##4##5{\trivlist
      \item[\hskip\labelsep{##4##1\ ##2}]{##4(##3)\@thmcounterend\ }\addcontentsline{toc}{subsection}{##1\ ##2\@thmcounterend}##5}}
\normalthmheadings
\def\reversethmheadings{\def\@spbegintheorem##1##2##3##4{\trivlist
                 \item[\hskip\labelsep{##3##2\ ##1\@thmcounterend}]##4\addcontentsline{toc}{subsection}{##2\ ##1\@thmcounterend}}
\def\@spopargbegintheorem##1##2##3##4##5{\trivlist
      \item[\hskip\labelsep{##4##2\ ##1}]{##4(##3)\@thmcounterend\ }\addcontentsline{toc}{subsection}{##2\ ##1\@thmcounterend}##5}}
\makeatother
\spnewtheorem*{remark*}{Remark}{\itshape}{\rmfamily}

\spnewtheorem*{example*}{Example}{\itshape}{\rmfamily}

\let\oldendproof\endproof
\def\endproof{\hfill\squareforqed\oldendproof}

\let\oldre\Re
\let\oldim\Im
\renewcommand{\Re}{\oldre\mathfrak{e}\,}
\renewcommand{\Im}{\oldim\mathfrak{m}\,}
\newcommand{\mathe}{\mathrm{e}}
\newcommand{\mathi}{\mathrm{i}}
\newcommand{\total}{\mathop{}\!\mathrm{d}}

\newcommand{\abs}[1]{{\left\lvert{#1}\right\rvert}}
\newcommand{\norm}[1]{{\left\lVert{#1}\right\rVert}}
\newcommand{\1}{\mathbbm{1}}
\newcommand{\eqend}[1]{\,#1}
\newcommand{\bigo}[1]{\mathcal{O}\left({#1}\right)}
\newcommand{\op}{\mathcal{O}}
\newcommand{\istest}{\in \mathcal{S}(\mathbb{R}^2)}

\makeatletter
\newcommand*{\bessel@n}[3]{\mathop{}\!\mathrm{#1}_{#2}\left(#3\right)}
\newcommand*{\bessel@s}[3]{\mathop{}\!\mathrm{#1}_{#2}\left[#3\right]}
\newcommand*{\bessel}{\@ifstar{\bessel@s}{\bessel@n}}
\makeatother

\DeclareMathOperator{\supp}{supp}
\DeclareMathOperator{\sgn}{sgn}

\usepackage{cite}

\usepackage[unicode=true,bookmarksopen=true,colorlinks=true,allcolors=blue,linktocpage=true,pdfa=true]{hyperref}

\allowdisplaybreaks

\usepackage{calc}
\journalname{Communications in Mathematical Physics}

\begin{document}

\title{A quantum energy inequality in the Sine--Gordon model}
\titlerunning{QEI in the Sine--Gordon model}

\author{Markus B. Fr{\"o}b \and Daniela Cadamuro}
\institute{Institut f{\"u}r Theoretische Physik, Universit{\"a}t Leipzig,\\ Br{\"u}derstra{\ss}e 16, 04103 Leipzig, Germany\\ \email{\href{mailto:mfroeb@itp.uni-leipzig.de}{mfroeb@itp.uni-leipzig.de}, \href{mailto:cadamuro@itp.uni-leipzig.de}{cadamuro@itp.uni-leipzig.de}}}
\authorrunning{M. B. Fr{\"o}b \and D. Cadamuro}

\date{14. December 2022, revised 11. May 2023}

\maketitle
\begin{abstract}
We consider the stress tensor in the massless Sine--Gordon model in the finite regime $\beta^2 < 4 \pi$ of the theory. We prove convergence of the renormalised perturbative series for the interacting stress tensor defined using the Bogoliubov formula in an arbitrary Hadamard state, even for the case that the smearing is only along a one-dimensional time-like worldline and not in space-time. We then show that the interacting energy density, as seen by an observer following this worldline, satisfies an absolute lower bound, that is a bound independent of the quantum state. Our proof employs and generalises existing techniques developed for free theories by Flanagan, Fewster and Smith.
\end{abstract}

\section{Introduction}
\label{sec_intro}

One of the fundamental observables both in classical and quantum field theory is the stress tensor (or stress-energy or energy-momentum tensor) $T_{\mu\nu}$. It enters the Einstein equations in General Relativity, linking the distribution of matter to the curvature of spacetime. In this context, the stress tensor should fulfill certain positivity conditions, such as the \emph{weak energy condition}, the \emph{strong energy condition}, or the \emph{null energy condition}, which imply constraints on exotic spacetime geometries. Assuming such conditions, various important results were derived such as the Penrose and Hawking singularity theorems~\cite{hawkingpenrose1970}, the Schoen--Yau positive mass theorem~\cite{schoenyau1979,schoenyau1981}, and Hawking's chronology protection results~\cite{hawking1992}.

However, in quantum field theory (QFT) these energy conditions are violated, and moreover one can find states in which the energy density at any single point can have arbitrarily negative expectation values~\cite{epsteinglaserjaffe1965,fewster2006}. Therefore, one might ask whether quantum fields are compatible with the energy conditions in the singularity theorems, or whether they allow for exotic spacetime geometries, such as wormholes. Moreover, unbounded negative energy densities have physically undesirable consequences regarding the validity of the second law of thermodynamics at macroscopic scales. To avoid these, there must be constraints on how negative energy densities of quantum fields can be. These constraints are called \emph{Quantum Energy Inequalities} (QEIs), and were first investigated by Ford~\cite{ford1991}. They are reminiscent of the classical energy conditions, suggesting that singularity or other theorems might still be obeyed by quantum fields, possibly in some weakened condition~\cite{brownfewsterkontou2018,fewsterkontou2020,freivogelkontoukrommydas2022,fewsterkontou2022}.

Technically, these bounds are of the following type: consider the energy density operator $T^{00}(t,\vec{x})$ averaged in time with the square of a real-valued smooth function:
\begin{equation}
T^{00}(f^2) \equiv \int T^{00}(t,\vec{x}) f^2(t) \total t \eqend{.}
\end{equation}
A QEI is an inequality of the form
\begin{equation}
\label{eq:masterqei}
\omega\left( T^{00}(f^2) \right) \geq - K_\omega(f)
\end{equation}
for all sufficiently regular quantum states $\omega$ of the system, where the constant $K_\omega$ depends on the smearing function $f$ and possibly on the state $\omega$. The QEI is called \emph{state-independent} or \emph{absolute} if $K$ is actually independent of $\omega$ and only depends on $f$. However, in many cases the dependence on the state is such that $K_\omega$ only depends mildly on the energy content of $\omega$~\cite{fewsterosterbrink2008}.

State-independent QEIs in the form~\eqref{eq:masterqei} have been established for \emph{free} quantum field theories in flat and curved spacetimes in a wide range of physical situations, including the free scalar, Dirac fermion, vector, and Rarita-Schwinger field, see the reviews~\cite{fewster2012,kontousanders2020} and references therein, as well as in conformal field theories (CFTs) in $1+1$-dimensions~\cite{fewsterhollands2005}. Moreover, they have been established also for smearings along a generic time-like worldline of an observer. In a more general setting, it was shown that analogues of QEIs are fulfilled by observables that arise from the operator product expansions of classically positive expressions~\cite{bostelmannfewster2009}. However, these observables do not necessarily have a direct physical interpretation (such as the energy density of a system).

These inequalities are also related to the Averaged Null Energy Condition (ANEC)~\cite{flanaganwald1996}, which is formally obtained from a QEI in the limit when the averaging function $f$ tends to $1$ and the integration is over a null geodesic. There is recent interest in the ANEC due to its connection to holography and the quantum information carried by black hole horizons, see for example~\cite{fukoellermarolf2017,boussoetal2016a,boussoetal2016b}.

However, the analysis of QEIs in theories with \emph{self-interaction} (apart from the case of CFTs) is a long-standing problem due to the complicated structure that local quantum fields have in the presence of non-trivial interactions. This is unsatisfactory: the dependence of the energy density on self-interaction is clearly of physical importance, and its lower bounds are related to the stability of spacetime, as explained above. A first result on QEIs in interacting massive theories was obtained by the second author, Bostelmann and Fewster~\cite{bostelmanncadamurofewster2013}, who derived a state-independent QEI for the massive Ising model. This model belongs to the class of integrable quantum field theories on $1+1$-dimensional Minkowski spacetime, and describes scalar bosons whose two-particle scattering matrix is given by $S_2 = - \1$. More generally, integrable models are characterised by a \emph{factorizing scattering marix}, where two-particle scattering processes determine the behaviour of particles completely. A proof of QEIs in a larger class of scalar integrable models (including the Sinh--Gordon model) was obtained in~\cite{bostelmanncadamuro2016}, on the level of expectation values of the energy density in one-particle states. Presently, the second author, Bostelmann and Mandrysch are also establishing QEIs in integrable models with momentum-independent two-particle scattering matrix in states of arbitrary particle number. Moreover, they are extending the analysis of QEIs at the one-particle level to integrable models with several particle species and with bound states. This includes several relevant examples, such as the Bullough-Dodd model, the $\mathrm{O}(N)$ nonlinear sigma model and the Federbush model. At the same time, the analysis extends to integrable models describing asymptotic fermions, first described in the non-perturbative algebraic quantum field theory (AQFT) framework by~\cite{bostelmanncadamuro2021}.

In this paper, we establish a state-independent or absolute QEI for the massless Sine--Gordon model in the finite regime of $\beta^2 < 4\pi$. For this, we use methods from the perturbative AQFT framework reviewed in~\cite{hollandswald2015,fewsterverch2015,fredenhagenrejzner2016}, in particular employing the Bogoliubov formula to define interacting field operators. To our knowledge, this is the first QEI that has been established in a self-interacting model without the use of any special symmetry (such as conformal symmetry) or the integrability of the model (like the Ising model as explained above). We build on our previous paper~\cite{froebcadamuro2022}, where we constructed the stress tensor $T_{\mu\nu}$ in the same model using an alternative approach based on the Gell--Mann-Low formula. There, we showed that the renormalised expectation value of the stress tensor is well-defined in the sense of distributions at each order in perturbation theory after removal of the infrared (IR) and ultraviolet (UV) cutoffs. Moreover, we proved that the renormalised perturbative series converges both for the Euclidean and the Lorentzian theory, for the latter in a large class of quasi-free Hadamard states. We also showed that the stress tensor is conserved after adding a quantum correction proportional to $\hbar$, yielding the expression
\begin{equation}
\label{eq:stresstensor}
\hat{T}_{\mu\nu} \equiv \op_{\mu\nu} - \frac{1}{2} \eta_{\mu\nu} \op_\rho{}^\rho + g \left( 1 - \frac{\beta^2}{8 \pi} \right) \eta_{\mu\nu} ( V_\beta + V_{-\beta} )
\end{equation}
with
\begin{equation}
\op_{\mu\nu} \equiv \partial_\mu \phi \, \partial_\nu \phi \eqend{.}
\end{equation}
Previously, various aspects of the Sine--Gordon model had also been studied in the pAQFT framework~\cite{bahnsrejzner2018,bahnsrejznerfredenhagen2021,bahnspinamontirejzner2021,martinschlesierzahn2023}.

The main results of the present work are as follows: As a first step, we show that the perturbative series for the renormalised interacting stress-energy tensor given by the Bogoliubov formula converges in expectation in an arbitrary quasi-free Hadamard state $\omega$. This requires almost the same conditions as in~\cite[Theorem~5]{froebcadamuro2022}, with only a slightly stronger assumption on the state-dependent part $W$ of the two-point function of $\omega$, namely the requirement of \emph{conditional positive semidefiniteness} explained in Remark~\ref{remark1} below.
\begin{theorem}[Renormalisation and convergence]
\label{thm_renorm}
Consider the massless Lorentzian Sine--Gordon model in the finite regime $\beta^2 < 4\pi$ and a quasi-free state $\omega^{\Lambda,\epsilon}$ in the vacuum sector whose two-point function has an IR cutoff $\Lambda$ and UV cutoff $\epsilon$. We furthermore require that the state-dependent part $W$ of the two-point function of the state $\omega$ satisfies:
\begin{enumerate}
\item $W(x,y)$ and its first and second derivatives grow at most polynomially,
\item $W(x,y)$ is conditionally positive semidefinite.
\end{enumerate}
Then there exists a redefinition of time-ordered products $\mathcal{T}\left[ \op_{\mu\nu}(x) \otimes V_\alpha(y) \right] \to \mathcal{T}\left[ \op_{\mu\nu}(x) \otimes V_\alpha(y) \right] + \delta \mathcal{T}\left[ \op_{\mu\nu}(x) \otimes V_\alpha(y) \right]$ with $\delta \mathcal{T}$ a local term proportional to $\delta^2(x-y) \mathcal{T}\left[ V_\alpha(y) \right]$, such that for any adiabatic cutoff function $g \istest$ the interacting modified stress tensor~\eqref{eq:stresstensor} defined by the Bogoliubov formula
\begin{equation}
\label{eq:interacting_stresstensor}
\mathcal{T}_\text{int}\left[ \hat{T}_{\mu\nu}(z) \right] \equiv \mathcal{T}\left[ \mathe_\otimes^{\mathi S_\text{int}} \right]^{\star (-1)} \star \mathcal{T}\left[ \hat{T}_{\mu\nu}(z) \otimes \mathe_\otimes^{\mathi S_\text{int}} \right]
\end{equation}
has a finite expectation value in the state $\omega$ in the physical limit $\Lambda, \epsilon \to 0$:
\begin{equation}
\lim_{\Lambda,\epsilon \to 0} \omega^{\Lambda,\epsilon}\left( \mathcal{T}_\text{int}\left[ \hat{T}_{\mu\nu}(f) \right] \right) < \infty \eqend{,}
\end{equation}
where the limit is taken termwise and $f \istest$. Moreover, the interacting modified stress tensor is conserved: we have
\begin{equation}
\label{eq:thm_renorm_cons_eq}
\lim_{\Lambda,\epsilon \to 0} \omega^{\Lambda,\epsilon}\left( \mathcal{T}_\text{int}\left[ \hat{T}_{\mu\nu}(\partial^\mu f) \right] \right) = 0
\end{equation}
for all $f \istest$ such that $g$ is constant on the support of $f$.
\end{theorem}
\begin{remark}
\label{remark1}
The condition of conditional positive semidefiniteness of $W$ is
\begin{equation}
\label{eq:thm_renorm_wpos_cont}
W(f,f^*) = \iint W(x,y) f(x) f^*(y) \total^2 x \total^2 y \geq 0
\end{equation}
for a complex test function $f \istest$ with vanishing mean $\int f(x) \total^2 x = 0$. This condition is clearly related to the positive definiteness of the state $\omega^{\Lambda,\epsilon}$ for derivatives of $\phi$, since
\begin{equation}
(\partial_\mu \phi)(f^\mu) = \int \partial_\mu \phi(x) f^\mu(x) \total^2 x = - \int \phi(x) \partial_\mu f^\mu(x) \total^2 x = - \phi(\partial_\mu f^\mu) \eqend{,}
\end{equation}
the test function $f(x) = \partial_\mu f^\mu(x)$ has vanishing mean, and positivity of the state in this case is
\begin{equation}
\label{eq:thm_renorm_pos_hplusw}
\omega^{\Lambda,\epsilon}\left( (\partial^\mu \phi)(f_\mu^*) \star (\partial^\nu \phi)(f_\nu) \right) = \mathi H^+\left( \partial^\mu f_\mu^*, \partial^\nu f_\nu \right) + W\left( \partial^\mu f_\mu^*, \partial^\nu f_\nu \right) > 0
\end{equation}
for $f_\mu \neq 0$, where $H^+$ is the Hadamard parametrix. Since $\mathi H^+$ is already positive definite for test functions with vanishing mean, the condition~\eqref{eq:thm_renorm_wpos_cont} is slightly stronger than would be strictly required, but seems to be necessary for the convergence of the perturbative series.

In our previous work~\cite{froebcadamuro2022}, we instead required the discrete condition
\begin{equation}
\label{eq:thm_renorm_wpos_disc}
\sum_{i,j=1}^n \Big[ W(x_i,x_j) - W(y_i,x_j) - W(x_i,y_j) + W(y_i,y_j) \Big] \geq 0
\end{equation}
for any configuration of points $x_i$ and $y_i$ and any $n \in \mathbb{N}$. However, this can be derived from~\eqref{eq:thm_renorm_wpos_cont} as follows: we take a sequence of test functions $\delta_N \istest$ approximating the $\delta$ distribution as $N \to \infty$, and $f$ to be the sum
\begin{equation}
f(x) = \sum_{i=1}^n \Big[ \delta^2_N(x-x_i) - \delta^2_N(x-y_i) \Big] \eqend{.}
\end{equation}
Taking the limit $N \to \infty$ after integration, we obtain~\eqref{eq:thm_renorm_wpos_disc} from~\eqref{eq:thm_renorm_wpos_cont}; such approximations and the existence of the limit for smooth $W$ are well known.

For the example given in our previous work (a state that is thermal in a wide range of energies) we have~\cite[Eq.~(14)]{froebcadamuro2022}
\begin{splitequation}
\label{eq:thm_renorm_w_example}
W(x,y) &= \int \mathe^{\mathi p (x-y)} \Theta\left( \abs{p^0} \in [E_0, E_1] \right) \frac{\pi \mathe^{- \beta \abs{p^0}}}{\abs{p^0} \left( 1 - \mathe^{- \beta \abs{p^0}} \right)} \\
&\quad\times [ \delta(p^1+p^0) + \delta(p^1-p^0) ] \frac{\total^2 p}{(2\pi)^2} \eqend{,}
\end{splitequation}
and since this is positive semidefinite for all $f \istest$ independently of their mean, it fulfills the condition~\eqref{eq:thm_renorm_wpos_cont}.
\end{remark}

Secondly, we show that the expectation value of the perturbative series for the renormalised interacting stress tensor in an arbitrary Hadamard state converges also when the smearing is only along a one-dimensional time-like worldline, and not in spacetime as in the previous theorem:
\begin{theorem}[Renormalisation of the quantum energy density]
\label{thm_renorm_qed}
We assume the same conditions as in Theorem~\ref{thm_renorm} and take a smearing function $f \in \mathcal{S}(\mathbb{R})$ supported on a smooth time-like curve $z^\mu(\tau)$, $\dot z^\mu(\tau) \dot z_\mu(\tau) = -1$, $\tau \in \mathbb{R}$. The quantum energy density
\begin{equation}
\label{eq:energy_density}
E_\omega(f) \equiv \lim_{\Lambda,\epsilon \to 0} \omega^{\Lambda,\epsilon}\left( \mathcal{T}_\text{int}\left[ \hat{T}_{\mu\nu}(F^{\mu\nu}) \right] \right)
\end{equation}
seen by an observer following the curve is finite, where
\begin{equation}
F^{\mu\nu}(z) \equiv \dot z^\mu(\tau) \dot z^\nu(\tau) f(\tau) \eqend{.}
\end{equation}
\end{theorem}

Finally, we show that the interacting energy density, as seen by an observer following a one-dimensional time-like worldline, satisfies a lower bound which is independent of the quantum state. Namely, we establish a state-independent QEI for the massless Sine--Gordon model:
\begin{theorem}[Quantum energy inequality]
\label{thm_qei}
We assume the same conditions as in Theorem~\ref{thm_renorm_qed}, but smear with the square of a test function $f \in \mathcal{S}(\mathbb{R})$. The smeared quantum energy density~\eqref{eq:energy_density} seen by an observer is then bounded from below independently of the quantum state $\omega$:
\begin{equation}
\label{eq:thm_qei}
E_\omega(f^2) \geq - K(z,f,\beta,g) \eqend{.}
\end{equation}
\end{theorem}
\begin{remark}
As we will see in the proof, there are three different contributions to the constant $K$ on the right-hand side of the inequality~\eqref{eq:thm_qei}, $K = K_0 + K_V + K_H$. The first one is the contribution from the free theory, which can be calculated exactly~\eqref{eq:qei_h_result}:
\begin{equation}
\label{eq:thm_qei_k0}
K_0 = \frac{1}{24 \pi} \int \left[ 6 \big[ f'(\tau) \big]^2 + f^2(\tau) \frac{[ \ddot z^1(\tau) ]^2}{1 + [ \dot z^1(\tau) ]^2} \right] \total \tau \eqend{.}
\end{equation}
It depends on the smearing function $f$ and the trajectory $z^\mu(\tau)$, but obviously not on $\beta$ nor $g$. It consists of two parts: the first one that arises for a straight trajectory with $z^\mu(\tau) = \delta^\mu_0 \tau$, $\ddot z^\mu = 0$ and coincides with old results~\cite{fewstereveson1998} in two-dimensional Minkowski space, and the second one which arises from the difference of the actual trajectory and a straight one, and consequently depends on the acceleration $\ddot z^\mu$.\footnote{We do not claim that this second term is new, but we haven't been able to find an explicit expression for it in the literature.} Moreover, $K_0$ stays finite in the limit where the trajectory becomes light-like $\dot z^1 \to \pm \infty$, at least if it is taken in such a way that $\ddot z^1/\dot z^1$ stays bounded. The second contribution $K_V$ is the one from the vertex operators $V_{\pm \beta}$ in the (quantum-corrected) stress tensor~\eqref{eq:stresstensor}, while the third one $K_H$ comes from derivatives of Hadamard parametrices that arise from the $\op_{\mu\nu}$ terms in the stress tensor~\eqref{eq:stresstensor}. Both of them clearly depend also on $\beta$ and $g$, but while $K_V$ is also finite in the light-like limit, the bounds that we derived to obtain $K_H$ diverge in this limit, see for example~\eqref{eq:mink_conv_ddf_bound} and the subsequent comments. The QEI~\eqref{eq:thm_qei} therefore only holds in this form for time-like trajectories $z^\mu(\tau)$, and we leave the derivation of a null QEI for the Sine--Gordon model to future work.

We do not believe our bound~\eqref{eq:thm_qei} to be optimal. Indeed, it is known~\cite{fewstereveson1998} that already the free-theory bound~\eqref{eq:thm_qei_k0} for a straight trajectory overestimates the optimal bound~\cite{flanagan1997} by a factor $3/2$. However, since we prove that the perturbative series converges, for sufficiently small coupling $g$ the contribution from interactions ($K_V$ and $K_H$) will be subdominant to the free-theory contribution from $K_0$. Unfortunately, through estimates such as~\eqref{eq:mink_conv_ddf_bound}, what is sufficiently small depends also on the trajectory $z^\mu(\tau)$ (except if it is straight).

Another interesting generalisation of our QEI~\eqref{eq:thm_qei} would be to consider non-Gaussian states, for example vector states obtained by applying smeared field operators on the vacuum. For such states, it is known that they are of Hadamard form, which in particular entails that the truncated $n$-point functions are smooth functions and that the two-point function has the same singular behaviour as the vacuum both in Minkowski and in globally hyperbolic curved spacetimes~\cite{hollandsruan2002,sanders2010}. Therefore, the renormalisation of the stress tensor~\cite{froebcadamuro2022} should proceed in the same way as for the Gaussian states that we consider here, and we expect that a similar QEI holds also in this case.
\end{remark}

\section{Proof of Theorem~\ref{thm_renorm} (Renormalisation and convergence)}
\label{sec_renorm}

We again use the framework of perturbative algebraic quantum field theory, but refer the reader to~\cite{hollandswald2015,fewsterverch2015,fredenhagenrejzner2016} for reviews and our previous work~\cite{froebcadamuro2022} for a short introduction and all relevant formulas. However, instead of defining interacting expectation values using the Gell-Mann--Low formula, we define interacting time-ordered products $\mathcal{T}_\text{int}$ (and interacting operators) using the Bogoliubov formula, and then take their expectation value.

Using that
\begin{equation}
\label{eq:t_starinverse}
\mathcal{T}\left[ \mathe_\otimes^{\mathi S_\text{int}} \right]^{\star (-1)} = \overline{\mathcal{T}}\left[ \mathe_\otimes^{ - \mathi S_\text{int} } \right]
\end{equation}
with the anti-time-ordered products $\overline{\mathcal{T}}$, the Bogoliubov formula defines interacting time-ordered products by
\begin{equation}
\label{eq:tint_def}
\mathcal{T}_\text{int}\left[ \op_1(f_1) \cdots \op_k(f_k) \right] \equiv \overline{\mathcal{T}}\left[ \mathe_\otimes^{ - \mathi S_\text{int} } \right] \star \mathcal{T}\left[ \mathe_\otimes^{ \mathi S_\text{int} } \otimes \op_1(f_1) \otimes \cdots \otimes \op_k(f_k) \right] \eqend{.}
\end{equation}
Here the operators $\op_i$ are classical polynomials of the basic field $\phi$ and its derivatives or vertex operators $V_\alpha \equiv \mathe^{\mathi \alpha \phi}$, smeared with test functions $f_i \istest$, and $S_\text{int}$ is the interaction with adiabatic cutoff. The time-ordered products $\mathcal{T}$ and anti-time-ordered products $\overline{\mathcal{T}}$ are constructed inductively, with the ones with single entries equal to the Hadamard-normal-ordered products
\begin{equation}
\label{eq:timeordered_single}
\mathcal{T}\left[ \op(x) \right] = \mathcal{N}_H\left[ \op(x) \right] = \overline{\mathcal{T}}\left[ \op(x) \right] \eqend{.}
\end{equation}
Outside the diagonal, the higher ones are defined by (anti-)causal factorisation:
\begin{splitequation}
\label{eq:timeordered_causal}
&\mathcal{T}\left[ \op_1(x_1) \otimes \cdots \otimes \op_n(x_n) \right] \\
&\quad= \mathcal{T}\left[ \op_1(x_1) \otimes \cdots \otimes \op_k(x_k) \right] \star \mathcal{T}\left[ \op_{k+1}(x_{k+1}) \otimes \cdots \otimes \op_n(x_n) \right]
\end{splitequation}
and
\begin{splitequation}
\label{eq:timeordered_anticausal}
&\overline{\mathcal{T}}\left[ \op_1(x_1) \otimes \cdots \otimes \op_n(x_n) \right] \\
&\quad= \overline{\mathcal{T}}\left[ \op_{k+1}(x_{k+1}) \otimes \cdots \otimes \op_n(x_n) \right] \star \overline{\mathcal{T}}\left[ \op_1(x_1) \otimes \cdots \otimes \op_k(x_k) \right]
\end{splitequation}
if none of the $x_1,\ldots,x_k$ lie in the past light cone of any of the $x_{k+1},\ldots,x_n$; the extension to the total diagonal $x_1 = \cdots = x_n$ corresponds to renormalisation. The interacting time-ordered products~\eqref{eq:tint_def} with a single entry define interacting operators in the quantum theory, and expanding the exponentials we obtain
\begin{equation}
\label{eq:opint_def}
\mathcal{T}_\text{int}\left[ \op(x) \right] = \sum_{n=0}^\infty \sum_{k=0}^n \frac{(-1)^k \mathi^n}{k! (n-k)!} \overline{\mathcal{T}}\left[ S_\text{int}^{\otimes k} \right] \star \mathcal{T}\left[ S_\text{int}^{\otimes (n-k)} \otimes \op(x) \right] \eqend{.}
\end{equation}

For the Sine--Gordon model, the interaction is given by
\begin{equation}
\label{eq:sint}
S_\text{int} = \int g(x) [ V_\beta(x) + V_{-\beta}(x) ] \total^2 x \eqend{,}
\end{equation}
where $g \istest$ is the adiabatic cutoff and we consider the case where $\beta^2 < 4 \pi$. We are interested in the interacting stress tensor, and from our previous work~\cite{froebcadamuro2022} we know that the classical expression needs to be amended to obtain a conserved tensor in the quantum theory. The corrected stress tensor is given by
\begin{equation}
\label{eq:stresstensor_def}
\hat{T}_{\mu\nu} \equiv \op_{\mu\nu} - \frac{1}{2} \eta_{\mu\nu} \op_\rho{}^\rho + g \left( 1 - \frac{\beta^2}{8 \pi} \right) \eta_{\mu\nu} ( V_\beta + V_{-\beta} )
\end{equation}
with the composite operator
\begin{equation}
\label{eq:opmunu_def}
\op_{\mu\nu} \equiv \partial_\mu \phi \, \partial_\nu \phi \eqend{,}
\end{equation}
such that we need to determine $\mathcal{T}_\text{int}\left[ \op_{\mu\nu} \right]$ and $\mathcal{T}_\text{int}\left[ V_{\pm \beta} \right]$. Moreover, we also need $\mathcal{T}_\text{int}\left[ \partial_\mu \phi \right]$ to prove the quantum energy inequality, so we first have to derive explicit expressions for these three interacting operators:
\begin{lemma}
\label{lemma1}
The interacting operators $\mathcal{T}_\text{int}\left[ V_\alpha(z) \right]$, $\mathcal{T}_\text{int}\left[ \partial_\mu \phi(z) \right]$ and $\mathcal{T}_\text{int}\left[ \op_{\mu\nu}(z) \right]$ are given by the formal power series~\eqref{eq:tint_vertex_def}, \eqref{eq:tint_1phi_def} and \eqref{eq:tint_2phi_def}.
\end{lemma}
\begin{proof}
We recall from~\cite{froebcadamuro2022} that the two-point function $G$ of the quasi-free state $\omega$ is given by
\begin{equation}
\label{eq:twopf_hadamard}
G^+(x,y) \equiv - \mathi \omega^{\Lambda,\epsilon}\left( \phi(x) \star \phi(y) \right) = H^+(x,y) + \frac{\mathi}{2 \pi} \ln\left( \frac{\Lambda}{\mu} \right) - \mathi W(x,y)
\end{equation}
with the Hadamard parametrix
\begin{equation}
\label{eq:hadamard}
H^+(x,y) \equiv \frac{\mathi}{4 \pi} \ln\left[ \mu^2 ( \epsilon + \mathi u ) ( \epsilon + \mathi v ) \right]
\end{equation}
containing the singular part of the two-point function and the state-dependent part $W(x,y)$ which is a real smooth and symmetric bisolution of the massless Klein--Gordon equation $\partial_x^2 W(x,y) = \partial_y^2 W(x,y) = 0$. We introduced the light cone coordinates
\begin{equation}
\label{eq:lightconecoords}
u = u(x,y) \equiv (x^0-y^0) - (x^1-y^1) \eqend{,} \quad v = v(x,y) \equiv (x^0-y^0) + (x^1-y^1) \eqend{,}
\end{equation}
$\Lambda$ is an IR cutoff which we ultimately take to vanish, and we keep $\epsilon > 0$ as a UV cutoff. That is, the physical two-point function is obtained as the distributional boundary value (in the limit $\epsilon \to 0$) from the function~\eqref{eq:twopf_hadamard} which is analytic for all $\epsilon > 0$.

From the proof of~\cite[Lemma 2]{froebcadamuro2022}, we know that
\begin{splitequation}
\label{eq:timeordered_exponential}
\mathcal{T}\left[ \bigotimes_{j=1}^n V_{\alpha_j}(x_j) \right] &= \exp\left[ - \mathi \sum_{1 \leq i < j \leq n} \alpha_i \alpha_j H^\text{F}(x_i,x_j) \right] \mathcal{N}_G\left[ \prod_{j=1}^n V_{\alpha_j}(x_j) \right] \\
&\quad\times \exp\left[ - \frac{1}{2} \sum_{i,j=1}^n \alpha_i \alpha_j W(x_i,x_j) \right] \left( \frac{\Lambda}{\mu} \right)^\frac{\left( \sum_{k=1}^n \alpha_k \right)^2}{4 \pi}
\end{splitequation}
and
\begin{splitequation}
\label{eq:timeordered_exponential_2phi}
&\mathcal{T}\left[ \bigotimes_{j=1}^n V_{\alpha_j}(x_j) \otimes \op_{\mu\nu}(z) \right] = \exp\left[ - \mathi \sum_{1 \leq i < j \leq n} \alpha_i \alpha_j H^\text{F}(x_i,x_j) \right] \\
&\quad\times \Bigg[ \sum_{i,j=1}^n \alpha_i \alpha_j \partial_\mu G_\text{F}(z,x_i) \partial_\nu G_\text{F}(z,x_j) \, \mathcal{N}_G\left[ \prod_{j=1}^n V_{\alpha_j}(x_j) \right] \\
&\qquad- 2 \sum_{i=1}^n \alpha_i \partial_{(\mu} G_\text{F}(z,x_i) \, \mathcal{N}_G\left[ \partial_{\nu)} \phi(z) \prod_{j=1}^n V_{\alpha_j}(x_j) \right] \\
&\qquad+ \mathcal{N}_G\left[ \op_{\mu\nu}(z) \prod_{j=1}^n V_{\alpha_j}(x_j) \right] + \lim_{z' \to z} \partial^z_\mu \partial^{z'}_\nu W(z,z') \, \mathcal{N}_G\left[ \prod_{j=1}^n V_{\alpha_j}(x_j) \right] \Bigg] \\
&\qquad\times \exp\left[ - \frac{1}{2} \sum_{i,j=1}^n \alpha_i \alpha_j W(x_i,x_j) \right] \left( \frac{\Lambda}{\mu} \right)^\frac{\left( \sum_{k=1}^n \alpha_k \right)^2}{4 \pi} \eqend{,} \raisetag{2.4em}
\end{splitequation}
where $\mathcal{N}_G$ denotes normal-ordering with respect to the full two-point function~\eqref{eq:twopf_hadamard}, $H^\text{F}$ is the Feynman parametrix
\begin{splitequation}
\label{eq:hf_def}
H^\text{F}(x,y) &\equiv \Theta(x^0-y^0) H^+(x,y) + \Theta(y^0-x^0) H^+(y,x) \\
&= \frac{\mathi}{4 \pi} \ln\left[ \mu^2 ( - u v + \mathi \epsilon \abs{u + v} + \epsilon^2 ) \right]
\end{splitequation}
which is a fundamental solution of the massless Klein--Gordon equation:
\begin{splitequation}
\label{eq:hf_fundamental_solution}
\partial^2 H^\text{F}(x,y) &= - 4 \partial_u \partial_v H^\text{F}(x,y) = \frac{2 \epsilon}{\pi ( u^2 + \epsilon^2 )} \delta(u+v) \\
&\to 2 \delta(u) \delta(v) = \delta^2(x-y) \quad (\epsilon \to 0) \eqend{,}
\end{splitequation}
and
\begin{equation}
\label{eq:twopf_feynman}
G^\text{F}(x,y) \equiv H^\text{F}(x,y) + \frac{\mathi}{2 \pi} \ln\left( \frac{\Lambda}{\mu} \right) - \mathi W(x,y)
\end{equation}
is the time-ordered two-point function or Feynman propagator. Completely analogous to the proof of~\cite[Lemma 2]{froebcadamuro2022}, we also compute
\begin{splitequation}
\label{eq:timeordered_exponential_1phi}
&\mathcal{T}\left[ \bigotimes_{j=1}^n V_{\alpha_j}(x_j) \otimes \phi(z) \right] = \exp\left[ - \mathi \sum_{1 \leq i < j \leq n} \alpha_i \alpha_j H^\text{F}(x_i,x_j) \right] \\
&\qquad\times\Bigg[ \mathcal{N}_G\left[ \phi(z) \prod_{j=1}^n V_{\alpha_j}(x_j) \right] - \sum_{i=1}^n \alpha_i G^\text{F}(x_i,z) \mathcal{N}_G\left[ \prod_{j=1}^n V_{\alpha_j}(x_j) \right] \Bigg] \\
&\qquad\times \exp\left[ - \frac{1}{2} \sum_{i,j=1}^n \alpha_i \alpha_j W(x_i,x_j) \right] \left( \frac{\Lambda}{\mu} \right)^\frac{\left( \sum_{k=1}^n \alpha_k \right)^2}{4 \pi} \eqend{,}
\end{splitequation}
as well as (using now the anti-causal factorisation~\eqref{eq:timeordered_anticausal})
\begin{splitequation}
\label{eq:antitimeordered_exponential}
\overline{\mathcal{T}}\left[ \bigotimes_{j=1}^n V_{\alpha_j}(x_j) \right] &= \exp\left[ - \mathi \sum_{1 \leq i < j \leq n} \alpha_i \alpha_j H^\text{D}(x_i,x_j) \right] \mathcal{N}_G\left[ \prod_{j=1}^n V_{\alpha_j}(x_j) \right] \\
&\quad\times \exp\left[ - \frac{1}{2} \sum_{i,j=1}^n \alpha_i \alpha_j W(x_i,x_j) \right] \left( \frac{\Lambda}{\mu} \right)^\frac{\left( \sum_{k=1}^n \alpha_k \right)^2}{4 \pi} \eqend{,}
\end{splitequation}
where
\begin{equation}
\label{eq:hd_def}
H^\text{D}(x,y) \equiv \Theta(x^0-y^0) H^+(y,x) + \Theta(y^0-x^0) H^+(x,y)
\end{equation}
is the anti-time-ordered or Dyson parametrix. As in~\cite[Lemma 2]{froebcadamuro2022}, these expressions are smooth functions of the $x_i$ if $\epsilon > 0$, but for $\epsilon = 0$ need renormalisation (which we will do later).

To determine the interacting operators that we need, we use the product formula~\cite[Eq.~(126)]{froebcadamuro2022}
\begin{equation}
\label{eq:normal_ordering_starproduct}
\mathcal{N}_G\left[ \mathe^{\mathi (J,\phi)} \right] \star \mathcal{N}_G\left[ \mathe^{\mathi (K,\phi)} \right] = \exp\left[ - \mathi (J, G^+ \ast K) \right] \, \mathcal{N}_G\left[ \mathe^{\mathi (J+K,\phi)} \right] \eqend{,}
\end{equation}
where $(J,\phi) \equiv \int J(x) \phi(x) \total^2 x$ and $(J, G^+ \ast K) \equiv \iint J(x) G^+(x,y) K(y) \total^2 x \total^2 y$. Taking functional derivatives thereof with respect to $J$ or $K$ and choosing $J$ and $K$ appropriately, we obtain the product for terms involving powers of $\phi$ and its derivatives, while for $J$ or $K$ proportional to sums of $\delta$ distributions we obtain the product for vertex operators. From the Bogoliubov formula~\eqref{eq:opint_def} and the Sine--Gordon interaction~\eqref{eq:sint}, we first compute
\begin{splitequation}
&\mathcal{T}_\text{int}\left[ \op(z) \right] = \sum_{n=0}^\infty \sum_{k=0}^n \frac{(-1)^k \mathi^n}{k! (n-k)!} \int\dotsi\int \sum_{\sigma_j = \pm 1} \overline{\mathcal{T}}\left[ \bigotimes_{i=1}^k V_{\sigma_i \beta}(x_i) \right] \\
&\qquad\qquad \star \mathcal{T}\left[ \bigotimes_{j=k+1}^n V_{\sigma_j \beta}(x_j) \otimes \op(z) \right] \prod_{j=1}^n g(x_j) \total^2 x_j \\
&\quad= \sum_{n=0}^\infty \sum_{k=0}^n \sum_{\ell=0}^k \sum_{m=0}^{n-k} \frac{(-1)^k \mathi^n}{\ell! (k-\ell)! m! (n-k-m)!} \int\dotsi\int \\
&\qquad\times \overline{\mathcal{T}}\left[ \bigotimes_{i=1}^\ell V_\beta(x_i) \bigotimes_{j=1}^{k-\ell} V_{-\beta}(y_j) \right] \star \mathcal{T}\left[ \bigotimes_{i=\ell+1}^{\ell+m} V_\beta(x_i) \bigotimes_{j=k-\ell+1}^{n-\ell-m} V_{-\beta}(y_j) \otimes \op(z) \right] \\
&\qquad\times \prod_{i=1}^{\ell+m} g(x_i) \total^2 x_i \prod_{j=1}^{n-\ell-m} g(y_j) \total^2 y_j \eqend{,}
\end{splitequation}
where we used that the (anti-)time-ordered products are symmetric in their arguments, and that there are $k!/[ \ell! (k-\ell)! ]$ ways to choose $\ell$ vertex operators $V_\beta$ with positive sign from a total of $k$ ones $V_{\pm \beta}$ with either sign; we also renamed the insertion points of the $V_{-\beta}$ operators to $y_j$.

Using then the above results for the (anti-)time-ordered products~\eqref{eq:antitimeordered_exponential} and \eqref{eq:timeordered_exponential_1phi} and the product formula~\eqref{eq:normal_ordering_starproduct} as well as the decomposition of the two-point function~\eqref{eq:twopf_hadamard}, a long but straightforward computation results in
\begin{splitequation}
\label{eq:tint_1phi_def}
\mathcal{T}_\text{int}\left[ \partial_\mu \phi(z) \right] &= - \sum_{n=0}^\infty \sum_{k=0}^n \sum_{\ell=0}^k \sum_{m=0}^{n-k} \frac{(-1)^k \mathi^n}{\ell! (k-\ell)! m! (n-k-m)!} \int\dotsi\int \mathcal{E}_{k,\ell,n,m}(\vec{x};\vec{y}) \\
&\quad\times\Bigg[ \beta \partial^z_\mu \mathcal{G}_{k,\ell,n,m}(\vec{x};\vec{y};z) \, \mathcal{N}_G\left[ \prod_{i=1}^{\ell+m} V_\beta(x_i) \prod_{j=1}^{n-\ell-m} V_{-\beta}(y_j) \right] \\
&\qquad\quad- \mathcal{N}_G\left[ \partial_\mu \phi(z) \prod_{i=1}^{\ell+m} V_\beta(x_i) \prod_{j=1}^{n-\ell-m} V_{-\beta}(y_j) \right] \Bigg] \\
&\quad\times \left( \frac{\Lambda}{\mu} \right)^\frac{\beta^2 (2(\ell+m)-n)^2}{4 \pi} \prod_{i=1}^{\ell+m} g(x_i) \total^2 x_i \prod_{j=1}^{n-\ell-m} g(y_j) \total^2 y_j \eqend{,} \raisetag{2.3em}
\end{splitequation}
where we defined
\begin{splitequation}
\label{eq:calg_def}
\mathcal{G}_{k,\ell,n,m}(\vec{x};\vec{y};z) &\equiv \sum_{i=1}^\ell H^+(x_i,z) + \sum_{i=\ell+1}^{\ell+m} H^\text{F}(x_i,z) - \sum_{j=1}^{k-\ell} H^+(y_j,z) \\
&\quad- \sum_{j=k-\ell+1}^{n-\ell-m} H^\text{F}(y_j,z) - \mathi \sum_{i=1}^{\ell+m} W(x_i,z) + \mathi \sum_{j=1}^{n-\ell-m} W(y_j,z)
\end{splitequation}
and
\begin{splitequation}
\label{eq:cale_def}
&\mathcal{E}_{k,\ell,n,m}(\vec{x};\vec{y}) \equiv \exp\Bigg[ - \mathi \beta^2 \sum_{1 \leq i < j \leq \ell} H^\text{D}(x_i,x_j) + \mathi \beta^2 \sum_{i=1}^\ell \sum_{j=1}^{k-\ell} H^\text{D}(x_i,y_j) \\
&\qquad- \mathi \beta^2 \sum_{1 \leq i < j \leq k-\ell} H^\text{D}(y_i,y_j) - \mathi \beta^2 \sum_{\ell+1 \leq i < j \leq \ell+m} H^\text{F}(x_i,x_j) \\
&\qquad+ \mathi \beta^2 \sum_{i=\ell+1}^{\ell+m} \sum_{j=k-\ell+1}^{n-\ell-m} H^\text{F}(x_i,y_j) - \mathi \beta^2 \sum_{k-\ell+1 \leq i < j \leq n-\ell-m} H^\text{F}(y_i,y_j) \\
&\qquad- \mathi \beta^2 \sum_{i=1}^\ell \sum_{j=\ell+1}^{\ell+m} H^+(x_i,x_j) + \mathi \beta^2 \sum_{i=1}^\ell \sum_{j=k-\ell+1}^{n-\ell-m} H^+(x_i,y_j) \\
&\qquad+ \mathi \beta^2 \sum_{i=1}^{k-\ell} \sum_{j=\ell+1}^{\ell+m} H^+(y_i,x_j) - \mathi \beta^2 \sum_{i=1}^{k-\ell} \sum_{j=k-\ell+1}^{n-\ell-m} H^+(y_i,y_j) \Bigg] \\
&\times \exp\left[ - \frac{\beta^2}{2} \sum_{i,j=1}^{\ell+m} W(x_i,x_j) + \beta^2 \sum_{i=1}^{\ell+m} \sum_{j=1}^{n-\ell-m} W(x_i,y_j) - \frac{\beta^2}{2} \sum_{i,j=1}^{n-\ell-m} W(y_i,y_j) \right] \eqend{.}
\end{splitequation}
Analogously, we obtain
\begin{splitequation}
\label{eq:tint_2phi_def}
&\mathcal{T}_\text{int}\left[ \op_{\mu\nu}(z) \right] = \sum_{n=0}^\infty \sum_{k=0}^n \sum_{\ell=0}^k \sum_{m=0}^{n-k} \frac{(-1)^k \mathi^n}{\ell! (k-\ell)! m! (n-k-m)!} \int\dotsi\int \mathcal{E}_{k,\ell,n,m}(\vec{x};\vec{y}) \\
&\quad\times \Bigg[ \beta^2 \partial^z_\mu \mathcal{G}_{k,\ell,n,m}(\vec{x};\vec{y};z) \partial^z_\nu \mathcal{G}_{k,\ell,n,m}(\vec{x};\vec{y};z) \, \mathcal{N}_G\left[ \prod_{i=1}^{\ell+m} V_\beta(x_i) \prod_{j=1}^{n-\ell-m} V_{-\beta}(y_j) \right] \\
&\qquad\quad- 2 \beta \partial^z_{(\mu} \mathcal{G}_{k,\ell,n,m}(\vec{x};\vec{y};z) \, \mathcal{N}_G\left[ \partial_{\nu)} \phi(z) \prod_{i=1}^{\ell+m} V_\beta(x_i) \prod_{j=1}^{n-\ell-m} V_{-\beta}(y_j) \right] \\
&\qquad\quad+ \mathcal{N}_G\left[ \op_{\mu\nu}(z) \prod_{i=1}^{\ell+m} V_\beta(x_i) \prod_{j=1}^{n-\ell-m} V_{-\beta}(y_j) \right] \\
&\qquad\quad+ \lim_{z' \to z} \partial^z_\mu \partial^{z'}_\nu W(z,z') \, \mathcal{N}_G\left[ \prod_{i=1}^{\ell+m} V_\beta(x_i) \prod_{j=1}^{n-\ell-m} V_{-\beta}(y_j) \right] \Bigg] \\
&\quad\times \left( \frac{\Lambda}{\mu} \right)^\frac{\beta^2 (2(\ell+m)-n)^2}{4 \pi} \prod_{i=1}^{\ell+m} g(x_i) \total^2 x_i \prod_{j=1}^{n-\ell-m} g(y_j) \total^2 y_j \raisetag{2.3em}
\end{splitequation}
and
\begin{splitequation}
\label{eq:tint_vertex_def}
&\mathcal{T}_\text{int}\left[ V_\alpha(z) \right] = \sum_{n=0}^\infty \sum_{k=0}^n \sum_{\ell=0}^k \sum_{m=0}^{n-k} \frac{(-1)^k \mathi^n}{\ell! (k-\ell)! m! (n-k-m)!} \int\dotsi\int \mathcal{E}_{k,\ell,n,m}(\vec{x};\vec{y}) \\
&\qquad\times \exp\left[ - \mathi \alpha \beta \, \mathcal{G}_{k,\ell,n,m}(\vec{x};\vec{y};z) - \frac{\alpha^2}{2} W(z,z) \right] \\
&\qquad\times \mathcal{N}_G\left[ \prod_{i=1}^{\ell+m} V_\beta(x_i) \prod_{j=1}^{n-\ell-m} V_{-\beta}(y_j) V_\alpha(z) \right] \\
&\qquad\times \left( \frac{\Lambda}{\mu} \right)^\frac{\left[ ( 2(\ell+m)-n ) \beta + \alpha \right]^2}{4 \pi} \prod_{i=1}^{\ell+m} g(x_i) \total^2 x_i \prod_{j=1}^{n-\ell-m} g(y_j) \total^2 y_j \eqend{.} \raisetag{2.3em}
\end{splitequation}
\end{proof}

The proofs of Theorems~\ref{thm_renorm} and~\ref{thm_renorm_qed} then proceeds analogously to the proof of~\cite[Thms.~4--6]{froebcadamuro2022} in our previous work, where we showed renormalisability and convergence of the Gell-Mann--Low series for the interacting expectation value of the stress tensor as well as its conservation. However, since now we are considering the Bogoliubov definition of the interacting stress tensor, as well as smearing functions that are supported on a smooth time-like curve (instead of space-time smearings), there are some small differences. To shorten the proof, we will only explain the differences and refer the reader for all other details to~\cite{froebcadamuro2022}.

As in~\cite{froebcadamuro2022}, we see that taking the expectation value in the quasi-free state $\omega^{\Lambda,\epsilon}$ with two-point function $G^+$ and taking the limit $\Lambda \to 0$, only neutral configurations contribute at each order in perturbation theory. For the (smeared) expectation value of the interacting modified stress tensor $\hat{T}_{\mu\nu}$, using the results~\eqref{eq:tint_2phi_def} and~\eqref{eq:tint_vertex_def}, we thus obtain the series
\begin{splitequation}
\label{eq:tint_tmunu_series}
&\lim_{\Lambda \to 0} \omega^{\Lambda,\epsilon}\left( \mathcal{T}_\text{int}\left[ \hat{T}_{\mu\nu}(f) \right] \right) = \sum_{n=0}^\infty \sum_{k=0}^n \sum_{\ell=0}^k \sum_{m=0}^{n-k} \frac{(-1)^k \mathi^n}{\ell! (k-\ell)! m! (n-k-m)!} \int\dotsi\int \\
&\qquad\times \mathcal{E}_{k,\ell,n,m}(\vec{x};\vec{y}) \Theta^{k,\ell,n,m}_{\mu\nu}(\vec{x};\vec{y};z) f(z) \total^2 z \prod_{i=1}^{\ell+m} g(x_i) \total^2 x_i \prod_{j=1}^{n-\ell-m} g(y_j) \total^2 y_j
\end{splitequation}
with
\begin{splitequation}
\label{eq:theta_munu_def}
&\Theta^{k,\ell,n,m}_{\mu\nu}(\vec{x};\vec{y};z) \equiv \delta_{2(\ell+m),n} \Bigg[ \beta^2 \partial^z_\mu \mathcal{G}_{k,\ell,n,m}(\vec{x};\vec{y};z) \partial^z_\nu \mathcal{G}_{k,\ell,n,m}(\vec{x};\vec{y};z) \\
&\hspace{6em}- \frac{\beta^2}{2} \eta_{\mu\nu} \partial_z^\rho \mathcal{G}_{k,\ell,n,m}(\vec{x};\vec{y};z) \partial^z_\rho \mathcal{G}_{k,\ell,n,m}(\vec{x};\vec{y};z) \\
&\hspace{6em}+ \lim_{z' \to z} \left( \partial^z_\mu \partial^{z'}_\nu W(z,z') - \frac{1}{2} \eta_{\mu\nu} \partial_z^\rho \partial^{z'}_\rho W(z,z') \right) \Bigg] \\
&\quad+ \left( 1 - \frac{\beta^2}{8 \pi} \right) \eta_{\mu\nu} g(z) \delta_{2(\ell+m),n-1} \exp\left[ - \mathi \beta^2 \mathcal{G}_{k,\ell,n,m}(\vec{x};\vec{y};z) - \frac{\beta^2}{2} W(z,z) \right] \\
&\quad+ \left( 1 - \frac{\beta^2}{8 \pi} \right) \eta_{\mu\nu} g(z) \delta_{2(\ell+m),n+1} \exp\left[ \mathi \beta^2 \mathcal{G}_{k,\ell,n,m}(\vec{x};\vec{y};z) - \frac{\beta^2}{2} W(z,z) \right] \eqend{.}
\end{splitequation}
Inserting the Wightman, Feynman and Dyson Hadamard parametrices~\eqref{eq:hadamard}, \eqref{eq:hf_def} and~\eqref{eq:hd_def} into $\mathcal{E}_{k,\ell,n,m}$~\eqref{eq:cale_def}, we obtain a product of terms of the form
\begin{equation}
\left[ \mu^2 ( - u(x_i,x_j) v(x_i,x_j) + \mathi \epsilon ) \right]^{\pm \frac{\beta^2}{4\pi}}
\end{equation}
as well as such terms with $x_i$ or $x_j$ replaced by $y_i$ or $y_j$, and terms with different $\mathi \epsilon$ prescriptions: $- \mathi \epsilon$ instead of $+ \mathi \epsilon$ for the Dyson parametrices, and $+ \mathi \epsilon \sgn(x_i^0 - x_j^0)$ for the Wightman parametrices. For a negative exponent, these are singular as $\epsilon \to 0$, but the singularity is absolutely integrable since we are in the finite regime $\beta^2 < 4\pi$. We can thus take their absolute value and ignore the $\mathi \epsilon$ prescription, and then these terms are bounded in the same way as in the proof of~\cite[Thm.~5]{froebcadamuro2022}. The same applies to the exponentials containing parametrices in $\Theta^{k,\ell,n,m}_{\mu\nu}$~\eqref{eq:theta_munu_def} which are contained in $\mathcal{G}_{k,\ell,n,m}$~\eqref{eq:calg_def}, and which come from the vertex operators in the stress tensor.

We thus only need to potentially renormalise the products of derivatives of $\mathcal{G}_{k,\ell,n,m}$~\eqref{eq:calg_def} in the first contribution to $\Theta^{k,\ell,n,m}_{\mu\nu}$~\eqref{eq:theta_munu_def}. Expanding those, we have three different types of terms: products of two $W$ which are smooth, products of $W$ and one parametrix, and products of two parametrices. For the products of $W$ and one parametrix, we integrate by parts to obtain
\begin{equation}
\int \partial_\mu W(x,z) \partial_\nu H(y,z) f(z) \total^2 z = - \int H(y,z) \partial_\nu \left[ \partial_\mu W(x,z) f(z) \right] \total^2 z \eqend{,}
\end{equation}
where the points $x$ and $y$ may be the same, and where the parametrix $H$ is either a Wightman one $H^+$ or a Feynman one $H^\text{F}$. Since $W$ and $f$ are smooth and the parametrix is only logarithmically divergent, the singularity is then integrable and we can bound this integral in the same way as in the proof of~\cite[Thm.~5]{froebcadamuro2022}. We now consider the terms with two parametrices. If they are supported at different points, i.e., for the terms of the form $\partial_\mu H(x,z) \partial_\nu H(y,z)$ with $x \neq y$, no renormalisation is needed and we can bound them in the same way as in the proof of~\cite[Thm.~5]{froebcadamuro2022}. For this, we use that the mixed derivatives can be rewritten in the form
\begin{splitequation}
\int \partial^z_{(u} H^\text{F}(x,z) \partial^z_{v)} H^\text{F}(y,z) f(z) \total^2 z &= \frac{1}{8} H^\text{F}(x,y) [ f(x) + f(y) ] \\
&\quad+ \frac{1}{2} \int H^\text{F}(x,z) H^\text{F}(y,z) \partial_u \partial_v f(z) \total^2 z \eqend{,}
\end{splitequation}
which is~\cite[Eq.~(187)]{froebcadamuro2022}, and which follows straightforwardly from an integration by parts and the fact that the Feynman parametrix $H^\text{F}$ is a fundamental solution of the massless Klein--Gordon equation~\eqref{eq:hf_fundamental_solution}. Since the Wightman parametrix $H^+$ is a solution, for their products we have instead
\begin{equation}
\int \partial^z_{(u} H^+(x,z) \partial^z_{v)} H^+(y,z) f(z) \total^2 z = \frac{1}{2} \int H^+(x,z) H^+(y,z) \partial_u \partial_v f(z) \total^2 z \eqend{,}
\end{equation}
and for the mixed products we obtain
\begin{splitequation}
\int \partial^z_{(u} H^+(x,z) \partial^z_{v)} H^\text{F}(y,z) f(z) \total^2 z &= \frac{1}{8} H^+(x,y) f(y) \\
&\quad+ \frac{1}{2} \int H^+(x,z) H^\text{F}(y,z) \partial_u \partial_v f(z) \total^2 z \eqend{.}
\end{splitequation}
On the other hand, for the same derivatives it is shown in~\cite[Eqs.~(195)--(198)]{froebcadamuro2022} that
\begin{equation}
\partial^z_u H^\text{F}(x,z) \partial^z_u H^\text{F}(y,z) = \frac{1}{32 \pi^2} \frac{\partial^2}{\partial u(z)^2} \ln\left[ \mu^2 u(z,a)^2 \right] \eqend{,}
\end{equation}
where $a$ is some point such that $\min(u(x),u(y)) \leq u(a) \leq \max(u(x),u(y))$, and the analogous result for $v$ derivatives, and the same computation establishes this result also for the Wightman and mixed parametrices. All these can now be bounded in the same way as in the proof of~\cite[Thm.~5]{froebcadamuro2022}.

On the other hand, the terms with products of two parametrices at the same point need renormalisation. We know that~\cite[Eq.~(150)]{froebcadamuro2022}
\begin{splitequation}
\label{eq:thm_renorm_hmuhnu}
\partial_\mu H^\text{F}(x,y) \partial_\nu H^\text{F}(x,y) &= \left[ \partial_\mu H^\text{F}(x,y) \partial_\nu H^\text{F}(x,y) \right]^\text{ren} \\
&\quad- \frac{\mathi}{4 \pi} \, \eta_{\mu\nu} \ln(2 \mu \epsilon) \, \delta^2(x-y) + \bigo{\epsilon} \eqend{,}
\end{splitequation}
with $\delta^2(x-y) = 2 \delta(u) \delta(v)$ and
\begin{equations}[eq:thm_renorm_hmuhnu_ren]
\left[ \partial_u H^\text{F}(x,y) \partial_u H^\text{F}(x,y) \right]^\text{ren} &= - \frac{\mathi}{4 \pi} \, \partial_u^2 H^\text{F}(x,y) - \frac{\mathi}{16 \pi} \delta^2(x-y) \eqend{,} \\
\left[ \partial_v H^\text{F}(x,y) \partial_v H^\text{F}(x,y) \right]^\text{ren} &= - \frac{\mathi}{4 \pi} \, \partial_v^2 H^\text{F}(x,y) - \frac{\mathi}{16 \pi} \delta^2(x-y) \eqend{,} \\
\left[ \partial_u H^\text{F}(x,y) \partial_v H^\text{F}(x,y) \right]^\text{ren} &= \frac{1}{2} \partial_u \partial_v \left[ H^\text{F}(x,y) \right]^2 \eqend{.}
\end{equations}
Moreover, in the proof of~\cite[Thm.~6]{froebcadamuro2022} it is shown that there exists a redefinition of time-ordered products~\cite[Eqs.~(220)--(224)]{froebcadamuro2022} such that the last term is removed from $\left[ \partial_u H^\text{F}(x,y) \partial_u H^\text{F}(x,y) \right]^\text{ren}$ and $\left[ \partial_v H^\text{F}(x,y) \partial_v H^\text{F}(x,y) \right]^\text{ren}$. It is also shown that there exists another redefinition of time-ordered products~\cite[Eqs.~(152)--(155)]{froebcadamuro2022} such that the divergent term $\sim \ln \epsilon$ is removed; however, this is actually not necessary for the stress tensor, since it anyway cancels out in $\Theta^{k,\ell,n,m}_{\mu\nu}$~\eqref{eq:theta_munu_def}. We now only need to determine the corresponding formulas for the product of two Wightman parametrices $H^+$, since the product of a Feynman parametrix and a Wightman one at the same point never appears. Namely, using the explicit form of the parametrix~\eqref{eq:hadamard}, we compute
\begin{splitequation}
\label{eq:thm_renorm_hadamard_uu_ren}
\partial_u H^+(x,y) \partial_u H^+(x,y) &= - \frac{1}{16 \pi^2} \frac{1}{(u - \mathi \epsilon)^2} = \frac{1}{16 \pi^2} \partial_u^2 \ln (u - \mathi \epsilon) \\
&= - \frac{\mathi}{4 \pi} \partial_u^2 H^+(x,y)
\end{splitequation}
and the analogous equation for $\partial_v H^+(x,y) \partial_v H^+(x,y)$, as well as
\begin{equation}
\label{eq:thm_renorm_hadamard_uv_ren}
\partial_u H^+(x,y) \partial_v H^+(x,y) = - \frac{1}{16 \pi^2} \frac{1}{(u - \mathi \epsilon) (v - \mathi \epsilon)} = \frac{1}{2} \partial_u \partial_v \Big[ H^+(x,y) \Big]^2
\end{equation}
for the mixed derivatives. For these terms we thus do not need any redefinition of time-ordered products, and can also bound them in the same way as the product of Feynman parametrices, as in the proof of~\cite[Thm.~5]{froebcadamuro2022}.

It then follows as in the proof of~\cite[Thm.~5]{froebcadamuro2022} that we have the bounds
\begin{splitequation}
\label{eq:thm_renorm_bound}
&\abs{ \int\dotsi\int \mathcal{E}_{k,\ell,n,m}(\vec{x};\vec{y}) \Theta^{k,\ell,n,m}_{\mu\nu}(\vec{x};\vec{y};z) f(z) \total^2 z \prod_{i=1}^{\ell+m} g(x_i) \total^2 x_i \prod_{j=1}^{n-\ell-m} g(y_j) \total^2 y_j } \\
&\leq \delta_{2(\ell+m),n} C n^2 K^n \left[ (\ell+m)! \right]^{1+\frac{\beta^2}{4 \pi}} + \delta_{2(\ell+m),n} C' (K')^n \left[ (\ell+m)! \right]^{1+\frac{\beta^2}{4 \pi}} \\
&\quad+ \left( \delta_{2(\ell+m),n-1} + \delta_{2(\ell+m),n+1} \right) C'' (K'')^{n+1} \left[ (\ell+m+1)! \right]^{1+\frac{\beta^2}{4 \pi}} \eqend{,} \raisetag{2em}
\end{splitequation}
with the constants $C$, $C'$, $C''$, $K$, $K'$ and $K''$ depending on $f$, $g$, $\beta$ and $W$, and using that
\begin{equation}
(\ell+m+1)^{1+\frac{\beta^2}{4 \pi}} \leq (\ell+m+1)^2
\end{equation}
for $\beta^2 < 4 \pi$, the expectation value of the interacting stress tensor~\eqref{eq:tint_tmunu_series} is bounded by
\begin{splitequation}
\label{eq:renorm_stress_bound}
&\abs{ \lim_{\Lambda,\epsilon \to 0} \omega^{\Lambda,\epsilon}\left( \mathcal{T}_\text{int}\left[ \hat{T}_{\mu\nu}(f) \right] \right)} \leq \hat{C} \sum_{n=0}^\infty \hat{K}^n \sum_{k=0}^n \sum_{\ell=0}^k \sum_{m=0}^{n-k} \frac{\left[ (\ell+m)! \right]^{1+\frac{\beta^2}{4 \pi}}}{\ell! (k-\ell)! m! (n-k-m)!} \\
&\quad\times \Big[ \delta_{2(\ell+m),n} (1+n^2) + \Big( \delta_{2(\ell+m),n-1} + \delta_{2(\ell+m),n+1} \Big) (\ell+m+1)^2 \Big] \eqend{,} \raisetag{1.6em}
\end{splitequation}
where $\hat{C} = \max(C,C',C'' K'')$ and $\hat{K} = \max(K,K',K'')$. It remains to estimate the sums, and we start with the first term where only terms with $2(\ell+m) = n$ contribute. Changing $n \to 2n$, we thus have to bound
\begin{splitequation}
\label{eq:renorm_stress_bound1}
&\hat{C} \sum_{n=0}^\infty [ 1 + (2n)^2 ] \hat{K}^{2n} \sum_{k=0}^{2n} \sum_{\ell=0}^k \sum_{m=0}^{2n-k} \frac{\left[ (\ell+m)! \right]^{1+\frac{\beta^2}{4 \pi}}}{\ell! (k-\ell)! m! (2n-k-m)!} \delta_{(\ell+m),n} \\
&\leq 2 \hat{C} \sum_{n=0}^\infty (n+1)^2 \hat{K}^{2n} (n!)^{1+\frac{\beta^2}{4 \pi}} \sum_{k=0}^{2n} \sum_{\ell=\max(0,k-n)}^{\min(k,n)} \frac{1}{\ell! (k-\ell)! (n-\ell)! (n-k+\ell)!} \eqend{.}
\end{splitequation}
We split the sum over $k$ in two and shift $k \to k+n$ in the second half, such that
\begin{splitequation}
&\sum_{k=0}^{2n} \sum_{\ell=\max(0,k-n)}^{\min(k,n)} \frac{1}{\ell! (k-\ell)! (n-\ell)! (n-k+\ell)!} \\
&= \sum_{k=0}^n \sum_{\ell=0}^k \frac{1}{\ell! (k-\ell)! (n-\ell)! (n-k+\ell)!} + \sum_{k=1}^n \sum_{\ell=k}^n \frac{1}{\ell! (k+n-\ell)! (n-\ell)! (\ell-k)!} \\
&= \sum_{\ell=0}^n \sum_{k=\ell}^n \frac{1}{\ell! (k-\ell)! (n-\ell)! (n-k+\ell)!} + \sum_{\ell=1}^n \sum_{k=1}^\ell \frac{1}{\ell! (k+n-\ell)! (n-\ell)! (\ell-k)!} \\
&= \sum_{\ell=0}^n \sum_{k=\ell}^n \frac{1}{\ell! (n-k)! (n-\ell)! k!} + \sum_{\ell=1}^n \sum_{k=0}^{\ell-1} \frac{1}{\ell! (n-k)! (n-\ell)! k!} \eqend{,} \raisetag{2.1em}
\end{splitequation}
where we switched the order of summation of both sums in the second equality, and replaced $k \to n-k+\ell$ in the first sum and $k \to \ell-k$ in the second sum in the third equality. We can then extend the summation in the second sum to include $\ell = 0$, and both sums combine to give
\begin{splitequation}
&\sum_{k=0}^{2n} \sum_{\ell=\max(0,k-n)}^{\min(k,n)} \frac{1}{\ell! (k-\ell)! (n-\ell)! (n-k+\ell)!} \\
&= \left[ \sum_{\ell=0}^n \frac{1}{\ell! (n-\ell)!} \right] \sum_{k=0}^n \frac{1}{(n-k)! k!} = \frac{2^{2n}}{(n!)^2} \eqend{.}
\end{splitequation}
It follows that~\eqref{eq:renorm_stress_bound1} is bounded by
\begin{equation}
4 \hat{C} \sum_{n=0}^\infty \frac{(n+1)^2 (2\hat{K})^{2n}}{(n!)^{1-\frac{\beta^2}{4 \pi}}} < \infty \eqend{,}
\end{equation}
since we are in the finite regime $\beta^2 < 4 \pi$.

For the second term in the sum~\eqref{eq:renorm_stress_bound} where only terms with $2(\ell+m) = n-1$ contribute, we change $n \to 2n+1$ and have to bound
\begin{splitequation}
\label{eq:renorm_stress_bound2}
&\hat{C} \sum_{n=0}^\infty \hat{K}^{2n+1} \sum_{k=0}^{2n+1} \sum_{\ell=0}^k \sum_{m=0}^{2n+1-k} \frac{\left[ (\ell+m)! \right]^{1+\frac{\beta^2}{4 \pi}}}{\ell! (k-\ell)! m! (2n+1-k-m)!} (\ell+m+1)^2 \delta_{(\ell+m),n} \\
&= \hat{C} \hat{K} \sum_{n=0}^\infty (n+1)^2 \hat{K}^{2n} \sum_{k=0}^{2n+1} \sum_{\ell=\max(0,k-n-1)}^{\min(k,n)} \frac{(n!)^{1+\frac{\beta^2}{4 \pi}}}{\ell! (k-\ell)! (n-\ell)! (n+1-k+\ell)!} \eqend{.}
\end{splitequation}
Similarly to before, we split the sum over $k$, shift $k$ in the second sum and switch the order of summation to obtain
\begin{splitequation}
&\sum_{k=0}^{2n+1} \sum_{\ell=\max(0,k-n-1)}^{\min(k,n)} \frac{1}{\ell! (k-\ell)! (n-\ell)! (n+1-k+\ell)!} \\
&= \sum_{\ell=0}^n \sum_{k=\ell}^n \frac{1}{\ell! (k-\ell)! (n-\ell)! (n+1-k+\ell)!} \\
&\quad+ \sum_{\ell=0}^n \sum_{k=0}^\ell \frac{1}{\ell! (k+n+1-\ell)! (n-\ell)! (\ell-k)!} \eqend{.}
\end{splitequation}
We now replace $k \to n+1-k+\ell$ in the first sum and $k \to \ell-k$ in the second sum, and obtain
\begin{splitequation}
&\sum_{k=0}^{2n+1} \sum_{\ell=\max(0,k-n-1)}^{\min(k,n)} \frac{1}{\ell! (k-\ell)! (n-\ell)! (n+1-k+\ell)!} \\
&= \left[ \sum_{\ell=0}^n \frac{1}{\ell! (n-\ell)!} \right] \sum_{k=0}^{n+1} \frac{1}{(n+1-k)! k!} = \frac{2^{2n+1}}{n! (n+1)!} \eqend{,}
\end{splitequation}
such that~\eqref{eq:renorm_stress_bound2} is bounded by
\begin{equation}
2 \hat{C} \hat{K} \sum_{n=0}^\infty \frac{(n+1) (2\hat{K})^{2n}}{(n!)^{1-\frac{\beta^2}{4 \pi}}} < \infty \eqend{.}
\end{equation}
The last term in the sum~\eqref{eq:renorm_stress_bound} is bounded analogously, such that the expectation value of the interacting stress tensor~\eqref{eq:renorm_stress_bound} has a finite bound; in particular, the perturbative sum from the Bogoliubov formula~\eqref{eq:opint_def} is absolutely convergent in this case.

Conservation of the interacting (modified) stress tensor then follows exactly as in the proof of~\cite[Thm.~6]{froebcadamuro2022}, and we omit the details. \hfill\squareforqed

\section{Proof of Theorem~\ref{thm_renorm_qed} (Renormalisation of the quantum energy density)}
\label{sec_renorm_qed}

While Theorem~\ref{thm_renorm} shows that the expectation value of the stress tensor is finite when smeared with a test function of space-time, for the QEI we need to smear the energy density with a test function supported on a time-like curve. We thus need to derive suitable bounds to show that also this smearing is finite.

We consider a smooth time-like curve $z^\mu(\tau)$ with $\tau \in \mathbb{R}$ and define the associated relativistic velocity $v^\mu(\tau) \equiv \dot z^\mu(\tau)$, which we assume to be normalised: $v^\mu v_\mu = -1$. We also define the associated space-like vector $w_\mu(\tau) \equiv - \epsilon_{\mu\nu} v^\nu(\tau)$, which is also normalised: $w^\mu w_\mu = 1$. Using that $\epsilon_{\mu\nu} \epsilon_{\rho\sigma} = - 2 \eta_{\mu[\rho} \eta_{\sigma]\nu}$, it follows that the Minkowski metric can be written as
\begin{equation}
\label{eq:metric_vv_ww}
\eta_{\mu\nu} = - v_\mu v_\nu + w_\mu w_\nu \eqend{,} \quad \eta^{\mu\nu} = - v^\mu v^\nu + w^\mu w^\nu \eqend{.}
\end{equation}
The energy density as seen by an observer following this curve is then given by the contraction of the stress tensor with $v^\mu v^\nu$, and the (smeared) expectation value of the quantum energy density is~\eqref{eq:energy_density}
\begin{equation}
E_\omega(f) \equiv \lim_{\Lambda,\epsilon \to 0} \omega^{\Lambda,\epsilon}\left( \mathcal{T}_\text{int}\left[ \hat{T}_{\mu\nu}(v^\mu v^\nu f) \right] \right)
\end{equation}
with $f \in \mathcal{S}(\mathbb{R})$, where we integrate over the one-dimensional submanifold $z(\tau)$.

Analogously to the previous result~\eqref{eq:tint_tmunu_series} for a space-time smearing, using the results~\eqref{eq:tint_2phi_def} and~\eqref{eq:tint_vertex_def} we obtain the series
\begin{splitequation}
\label{eq:eomega_series}
E_\omega(f) &= \sum_{n=0}^\infty \sum_{k=0}^n \sum_{\ell=0}^k \sum_{m=0}^{n-k} \frac{(-1)^k \mathi^n}{\ell! (k-\ell)! m! (n-k-m)!} \int\dotsi\int \mathcal{E}_{k,\ell,n,m}(\vec{x};\vec{y}) \\
&\qquad\times \Theta^{k,\ell,n,m}(\vec{x};\vec{y};\tau) f(\tau) \total \tau \prod_{i=1}^{\ell+m} g(x_i) \total^2 x_i \prod_{j=1}^{n-\ell-m} g(y_j) \total^2 y_j
\end{splitequation}
for the (smeared) expectation value of the interacting energy density, where
\begin{splitequation}
\label{eq:theta_00_def}
&\Theta^{k,\ell,n,m}(\vec{x};\vec{y};\tau) \equiv v^\mu(\tau) v^\nu(\tau) \Theta^{k,\ell,n,m}_{\mu\nu}(\vec{x};\vec{y};z(\tau)) \\
&\quad= \frac{1}{2} \left( v^\mu v^\nu + w^\mu w^\nu \right) \delta_{2(\ell+m),n} \Bigg[ \beta^2 \partial^z_\mu \mathcal{G}_{k,\ell,n,m}(\vec{x};\vec{y};z) \partial^z_\nu \mathcal{G}_{k,\ell,n,m}(\vec{x};\vec{y};z) \\
&\hspace{14em}+ \lim_{z' \to z} \partial^z_\mu \partial^{z'}_\nu W(z,z') \Bigg] \\
&\qquad- \left( 1 - \frac{\beta^2}{8 \pi} \right) g(z) \delta_{2(\ell+m),n-1} \exp\left[ - \mathi \beta^2 \mathcal{G}_{k,\ell,n,m}(\vec{x};\vec{y};z) - \frac{\beta^2}{2} W(z,z) \right] \\
&\qquad- \left( 1 - \frac{\beta^2}{8 \pi} \right) g(z) \delta_{2(\ell+m),n+1} \exp\left[ \mathi \beta^2 \mathcal{G}_{k,\ell,n,m}(\vec{x};\vec{y};z) - \frac{\beta^2}{2} W(z,z) \right] \eqend{,}
\end{splitequation}
and we evaluate it on the observer's worldline $z = z(\tau)$. To simplify the expression for $\Theta^{k,\ell,n,m}$, we have used the expression for the inverse metric in terms of $v^\mu$ and $w^\mu$~\eqref{eq:metric_vv_ww}, from which it follows that for any $A_\mu$ and $B_\mu$
\begin{equation}
v^\mu v^\nu \left( A_\mu B_\nu - \frac{1}{2} \eta_{\mu\nu} A^\rho B_\rho \right) = \frac{1}{2} \left( v^\mu v^\nu + w^\mu w^\nu \right) A_\mu B_\nu \eqend{.}
\end{equation}
To show that $E_\omega(f)$ is renormalisable and the series~\eqref{eq:eomega_series} convergent, we need to show bounds on the smearing in $\tau$ which are of the same form as for a space-time smearing. Afterwards, we can proceed exactly as in the proofs of Theorem~\ref{thm_renorm} (the previous section) and~\cite[Thms.~4 and 5]{froebcadamuro2022}, and we omit the details.

For the required bounds, we first consider the last two terms in~\eqref{eq:theta_00_def}, which contain exponentials of the parametrix. As in the previous section, inserting the various parametrices into $\mathcal{E}_{k,\ell,n,m}$ and $\mathcal{G}_{k,\ell,n,m}$, we obtain after taking the absolute value a product of terms of the form
\begin{equation}
\abs{ \mu^2 u(x_i,x_j) v(x_i,x_j) }^{\pm \frac{\beta^2}{4\pi}}
\end{equation}
as well as such terms with $x_i$ or $x_j$ replaced by $y_i$ or $y_j$ or $z$. As in the proof of~\cite[Thm.~5]{froebcadamuro2022}, we can write them using the Cauchy determinant formula as a determinant, which can then be estimated by the absolute value of the sum of individual terms. For this, we need to combine the exponentials $\exp\left[ \pm \mathi \beta^2 \mathcal{G}_{k,\ell,n,m}(\vec{x};\vec{y};z) - \beta^2/2 W(z,z) \right]$ with $\mathcal{E}_{k,\ell,n,m}(\vec{x};\vec{y})$ analogously to~\cite[Eq.~(202)]{froebcadamuro2022}, renaming $z$ to $x_{\ell+m+1}$ for the second-to-last term in~\eqref{eq:theta_00_def} or $y_{n-\ell-m+1}$ for the last term in~\eqref{eq:theta_00_def}. Since both cases are bounded in the same way, we only show the first explicitly. We thus need to bound
\begin{splitequation}
\label{eq:thm_renorm_vertex_bound}
&\int\dotsi\int \left[ \sum_{\pi} \left( \prod_{j=1}^k \frac{1}{\mu \abs{u(x_j,y_{\pi(j)})}} \right) \frac{1}{\mu \abs{u(z(\tau),y_{\pi(k+1)})}} \right]^\frac{\beta^2}{4 \pi} \\
&\qquad\times \left[ \sum_{\pi} \left( \prod_{j=1}^k \frac{1}{\mu \abs{v(x_j,y_{\pi(j)})}} \right) \frac{1}{\mu \abs{v(z(\tau),y_{\pi(k+1)})}} \right]^\frac{\beta^2}{4 \pi} \\
&\qquad\times \abs{g(z(\tau))} \abs{f(\tau)} \total \tau \prod_{i=1}^k \abs{ g(x_i) } \total u(x_i) \total v(x_i) \prod_{i=1}^{k+1} \abs{ g(y_i) } \total u(y_i) \total v(y_i) \eqend{,}
\end{splitequation}
where $k = \ell+m$ and the sums run over all permutations $\pi$ of $\{1,\ldots,k+1\}$. The only difference to~\cite[Eq.~(163)]{froebcadamuro2022} is then that instead of a space-time integral over $x_{k+1}$, we only have a one-dimensional integral over $\tau$. We thus proceed exactly as in the proof of~\cite[Thm.~5]{froebcadamuro2022}, producing IR convergence factors $\big[ 1 + \mu^2 u(x)^2 \big]^{-2} \big[ 1 + \mu^2 v(x)^2 \big]^{-2}$ by bounding their inverse together with the adiabatic cutoff functions $g(x)$, using Hölder's inequality to bound the resulting integrals which factorise in $u$ and $v$, and extracting the sum over permutations (which gives an additional factor of $[(k+1)!]^{2/\rho}$ with $\rho = 8 \pi/(4 \pi + \beta^2)$). The only new bound that we need to prove is a bound on
\begin{splitequation}
&\iiint \Big[ \mu^2 \abs{u(z(\tau),y_{k+1}) v(z(\tau),y_{k+1})} \Big]^{- \rho \frac{\beta^2}{4 \pi}} \abs{f(\tau)} \total \tau \\
&\quad\times \frac{\total u(y_{k+1})}{[ 1 + \mu^2 u(y_{k+1})^2 ]^\rho} \frac{\total v(y_{k+1})}{[ 1 + \mu^2 v(y_{k+1})^2 ]^\rho} \eqend{,}
\end{splitequation}
since the integral over $x_{k+1} = z(\tau)$ is now only one-dimensional. We shift $u(y_{k+1}) \to u(y_{k+1}) + u(z(\tau))$, $v(y_{k+1}) \to v(y_{k+1}) + v(z(\tau))$ and perform these integrals first. We split the integral over $u(y_{k+1})$ into two parts: If $\mu \abs{u(y_{k+1})} \leq 1$, we estimate the factor $[ 1 + \mu^2 [ u(y_{k+1}) + u(z(\tau)) ]^2 ]^{-\rho}$ by $1$ and obtain
\begin{equation}
\int_{\mu \abs{u(y_{k+1})} \leq 1} \Big[ \mu \abs{u(y_{k+1})} \Big]^{- \rho \frac{\beta^2}{4 \pi}} \frac{\total u(y_{k+1})}{[ 1 + \mu^2 [ u(y_{k+1}) + u(z(\tau)) ]^2 ]^\rho} \leq \frac{8 \pi}{4 \pi - \rho \beta^2} \mu^{-1} \eqend{,}
\end{equation}
while for $\mu \abs{u(y_{k+1})} > 1$ we estimate
\begin{splitequation}
&\int_{\mu \abs{u(y_{k+1})} > 1} \Big[ \mu \abs{u(y_{k+1})} \Big]^{- \rho \frac{\beta^2}{4 \pi}} \frac{\total u(y_{k+1})}{[ 1 + \mu^2 [ u(y_{k+1}) + u(z(\tau)) ]^2 ]^\rho} \\
&\quad\leq \int \frac{\total u(y_{k+1})}{[ 1 + \mu^2 [ u(y_{k+1}) + u(z(\tau)) ]^2 ]^\rho} = \frac{\sqrt{\pi}}{\mu} \frac{\Gamma\left( \rho - \frac{1}{2} \right)}{\Gamma(\rho)} \eqend{.}
\end{splitequation}
The integral over $v(y_{k+1})$ is bounded in the same way, and finally the integral over $\tau$ gives $\norm{ f }_1$. It follows that for the contribution of the last two terms in~\eqref{eq:theta_00_def} to the expectation value of the energy density we have a bound of the same form as for a space-time smearing, namely
\begin{splitequation}
\label{eq:thm_renorm_vertex_bound2}
&\int\dotsi\int \abs{ \mathcal{E}_{k,\ell,n,m}(\vec{x};\vec{y}) } \left( 1 - \frac{\beta^2}{8 \pi} \right) \delta_{2(\ell+m),n+1} \abs{ g(z) } \abs{ f(\tau) } \\
&\times \abs{ \exp\left[ \mathi \beta^2 \mathcal{G}_{k,\ell,n,m}(\vec{x};\vec{y};z) - \frac{\beta^2}{2} W(z,z) \right] } \total \tau \prod_{i=1}^{\ell+m} \abs{ g(x_i) } \total^2 x_i \prod_{j=1}^{n-\ell-m} \abs{ g(y_j) } \total^2 y_j \\
&\leq \delta_{2(\ell+m),n-1} C K^{n+1} \left[ (\ell+m+1)! \right]^{1+\frac{\beta^2}{4 \pi}} \raisetag{1.4em}
\end{splitequation}
with constants $C$ and $K$ depending on $f$, $g$, $\beta$ and $z(\tau)$; compare~\eqref{eq:thm_renorm_bound} for the space-time smearing. We note that due to our condition~\eqref{eq:thm_renorm_wpos_disc}
\begin{equation}
\sum_{i,j=1}^n \Big[ W(x_i,x_j) - W(x_i,y_j) - W(y_i,x_j) + W(y_i,y_j) \Big] \geq 0
\end{equation}
for all $n \in \mathbb{N}$, the $W$ dependence in the exponential in $\mathcal{E}_{k,\ell,n,m}(\vec{x};\vec{y})$~\eqref{eq:cale_def}, combined with $W(z,z)$, is simply estimated by $1$, and thus the constants $C$ and $K$ do not depend on $W$.

To bound the first term in~\eqref{eq:theta_00_def}, we first consider the second derivative of the state-dependent part $W$. A straightforward bound is
\begin{splitequation}
&\abs{ \int \left. \frac{1}{2} \left( v^\mu v^\nu + w^\mu w^\nu \right) \lim_{z' \to z} \partial^z_\mu \partial^{z'}_\nu W(z,z') \right\rvert_{z = z(\tau)} f(\tau) \total \tau } \\
&\quad\leq \sup_{\mu \in \{0,1\} } \norm{ \sup(v^\mu) }_\infty^2 \sup_{\mu,\nu \in \{0,1\}} \norm{ \partial^z_\mu \partial^{z'}_\nu W(z,z') }_\infty \norm{ f }_1 \eqend{,}
\end{splitequation}
where we used that $w^1 = v^0$, $w^0 = v^1$. For the terms contained in $\mathcal{G}_{k,\ell,n,m}$~\eqref{eq:calg_def}, we again have three different types of terms: products of two $W$, products of $W$ and one parametrix, and products of two parametrices. For products of two $W$, we again have a straightforward bound
\begin{splitequation}
&\abs{ \int \left. \frac{1}{2} \left( v^\mu v^\nu + w^\mu w^\nu \right) \partial^z_\mu W(x,z) \partial^z_\nu W(y,z) \right\rvert_{z = z(\tau)} f(\tau) \total \tau } \\
&\quad\leq \sup_{\mu \in \{0,1\} } \norm{ \sup(v^\mu) }_\infty^2 \sup_{\mu \in \{0,1\}} \norm{ \partial^z_\mu W(x,z) }_\infty^2 \norm{ f }_1 \eqend{,}
\end{splitequation}
where $x$ and $y$ can be equal or distinct. For the products of $W$ and one parametrix, we compute instead
\begin{splitequation}
\label{eq:thm_renorm_qed_whder}
&\int \left. \frac{1}{2} \left( v^\mu v^\nu + w^\mu w^\nu \right) \partial^z_\mu W(x,z) \partial^z_\nu H(y,z) \right\rvert_{z = z(\tau)} f(\tau) \total \tau \\
&= \int \Big[ (v^0 - v^1)^2 \partial_{u(z)} W(x,z) \partial_{u(z)} H(y,z) \\
&\qquad+ (v^0 + v^1)^2 \partial_{v(z)} W(x,z) \partial_{v(z)} H(y,z) \Big]_{z = z(\tau)} f(\tau) \total \tau \eqend{,}
\end{splitequation}
using that
\begin{equation}
v^u = v^0 - v^1 \eqend{,} \quad v^v = v^0 + v^1 \eqend{,} \quad w^u = - v^u \eqend{,} \quad w^v = v^v \eqend{.}
\end{equation}
In particular, the mixed terms where one derivative is with respect to $u$ and one with respect to $v$ have dropped out. We now need to change the derivatives acting on the parametrix into the $\tau$ derivative, for which we first compute
\begin{equations}[eq:dudv_dtau]
\partial_\tau u(z(\tau)) &= \partial_\tau \Big[ z^0(\tau) - z^1(\tau) \Big] = v^0 - v^1 = \sqrt{1 + (v^1)^2} - v^1 > 0 \eqend{,} \\
\partial_\tau v(z(\tau)) &= \partial_\tau \Big[ z^0(\tau) + z^1(\tau) \Big] = v^0 + v^1 = \sqrt{1 + (v^1)^2} + v^1 > 0 \eqend{,}
\end{equations}
using that $z(\tau)$ is a future-directed time-like curve. For the Wightman parametrix \eqref{eq:hadamard}, it follows that
\begin{splitequation}
\label{eq:hadamard_der_u_tau}
\partial_{u(z)} H^+(y,z) \Big\rvert_{z = z(\tau)} &= - \frac{\mathi}{4 \pi} \frac{1}{u(y,z) - \mathi \epsilon} \bigg\rvert_{z = z(\tau)} \\
&= \frac{\mathi}{4 \pi} \frac{1}{v^0 - v^1} \partial_\tau \ln\left[ \mu ( \epsilon + \mathi u(y,z(\tau)) ) \right]
\end{splitequation}
and the analogous equation with $v^0 + v^1$ for the $v(z)$ derivative, while for the Feynman parametrix~\eqref{eq:hf_def} we obtain
\begin{splitequation}
\label{eq:hf_der_u_tau}
&\partial_{u(z)} H^\text{F}(y,z) \Big\rvert_{z = z(\tau)} = - \frac{\mathi}{4 \pi} \left[ \frac{\Theta(u(y,z)+v(y,z))}{u(y,z) - \mathi \epsilon} + \frac{\Theta(-u(y,z)-v(y,z))}{u(y,z) + \mathi \epsilon} \right]_{z = z(\tau)} \\
&\quad= \frac{\mathi}{4 \pi} \frac{1}{v^0 - v^1} \partial_\tau \Big[ \Theta(u(y,z(\tau))+v(y,z(\tau))) \ln\left[ \mu ( \epsilon + \mathi u(y,z(\tau)) ) \right] \\
&\hspace{8em}+ \Theta(-u(y,z(\tau))-v(y,z(\tau))) \ln\left[ \mu ( \epsilon - \mathi u(y,z(\tau)) ) \right] \Big] \\
&\qquad- \frac{\mathi}{4 \pi} \frac{v^0}{v^0 - v^1} \delta(y^0-z^0(\tau)) \mathi \pi \sgn[ y^1 - z^1(\tau) ] \eqend{,} \raisetag{1.6em}
\end{splitequation}
and the analogous equation with $u$ replaced by $v$, $v^0 - v^1$ by $v^0 + v^1$ and the $\sgn$ function replaced by its negative. Integrating the $\tau$ derivative by parts, and changing coordinates to either $u(z(\tau))$ or $v(z(\tau))$ according to the argument of the logarithm, for the Wightman parametrix we thus obtain
\begin{splitequation}
\label{eq:thm_renorm_qed_bound_wh1}
&\int \left. \frac{1}{2} \left( v^\mu v^\nu + w^\mu w^\nu \right) \partial^z_\mu W(x,z) \partial^z_\nu H^+(y,z) \right\rvert_{z = z(\tau)} f(\tau) \total \tau \\
&= - \frac{\mathi}{4 \pi} \int \bigg[ f(\tau) \partial_\tau \Big[ \partial_{u(z)} W(x,z) \Big]_{z = z(\tau)} + f(\tau) \partial_\tau \ln(v^0 - v^1) \Big[ \partial_{u(z)} W(x,z) \Big]_{z = z(\tau)} \\
&\qquad\qquad+ \partial_\tau f(\tau) \Big[ \partial_{u(z)} W(x,z) \Big]_{z = z(\tau)} \bigg] \ln\left[ \mu ( \epsilon + \mathi u(y,z(\tau)) ) \right] \total u(z(\tau)) \\
&\quad- \frac{\mathi}{4 \pi} \int \bigg[ f(\tau) \partial_\tau \Big[ \partial_{v(z)} W(x,z) \Big]_{z = z(\tau)} + f(\tau) \partial_\tau \ln (v^0 + v^1) \Big[ \partial_{v(z)} W(x,z) \Big]_{z = z(\tau)} \\
&\qquad\qquad+ \partial_\tau f(\tau) \Big[ \partial_{v(z)} W(x,z) \Big]_{z = z(\tau)} \bigg] \ln\left[ \mu ( \epsilon + \mathi v(y,z(\tau)) ) \right] \total v(z(\tau)) \eqend{,}
\end{splitequation}
where now $\tau$ (and thus also $z^\mu(\tau)$) is a function of $u(z(\tau))$ in the first integral and of $v(z(\tau))$ in the second integral. This identification proceeds through the identities
\begin{equation}
\label{eq:thm_renorm_udot}
u(\tau) = z^0(\tau) - z^1(\tau) \eqend{,} \quad \dot u(\tau) = v^0(\tau) - v^1(\tau) = \sqrt{1 + [ v^1(\tau) ]^2} - v^1(\tau) \eqend{,}
\end{equation}
where we used the normalisation $v^\mu v_\mu = -1$. It follows that
\begin{equation}
v^1(\tau) = \frac{1 - \dot u(\tau)^2}{2 \dot u(\tau)} \eqend{,} \quad v^0(\tau) = \frac{1 + \dot u(\tau)^2}{2 \dot u(\tau)} \eqend{,}
\end{equation}
where we used that $z^\mu(\tau)$ is a future-directed curve such that $v^0 > 0$ and $\dot u(\tau) > 0$, and from this finally
\begin{equation}
\label{eq:zmu_tau_of_uv}
z^\mu(\tau) = z^\mu(0) + \int_0^\tau v^\mu(\sigma) \total \sigma \eqend{,} \quad v(z(\tau)) = v(u(\tau)) = v(0) + \int_0^\tau \frac{1}{\dot u(\sigma)} \total \sigma \eqend{.}
\end{equation}
Analogously, for the second integral we have to use
\begin{equation}
u(z(\tau)) = u(v(\tau)) = u(0) + \int_0^\tau \frac{1}{\dot v(\sigma)} \total \sigma \eqend{.}
\end{equation}
To bound~\eqref{eq:thm_renorm_qed_bound_wh1}, we estimate the terms in brackets by their supremum:
\begin{splitequation}
&\bigg\lvert f(\tau) \partial_\tau \Big[ \partial_{u(z)} W(x,z) \Big]_{z = z(\tau)} + f(\tau) \partial_\tau \ln(v^0 - v^1) \Big[ \partial_{u(z)} W(x,z) \Big]_{z = z(\tau)} \\
&\quad+ \partial_\tau f(\tau) \Big[ \partial_{u(z)} W(x,z) \Big]_{z = z(\tau)} \bigg\rvert \\
&\leq \frac{1}{1 + \mu^2 u(z(\tau))^2} \sup_\tau \bigg\lvert [ 1 + \mu^2 u(z(\tau))^2 ] \bigg[ f(\tau) \partial_\tau \Big[ \partial_{u(z)} W(x,z) \Big]_{z = z(\tau)} \\
&\quad+ f(\tau) \partial_\tau \ln(v^0 - v^1) \Big[ \partial_{u(z)} W(x,z) \Big]_{z = z(\tau)} + \partial_\tau f(\tau) \Big[ \partial_{u(z)} W(x,z) \Big]_{z = z(\tau)} \bigg] \bigg\rvert
\end{splitequation}
and the analogous estimate for the second integral, as well as
\begin{equation}
\abs{ \ln\left[ \mu ( \epsilon + \mathi u(y,z(\tau)) ) \right] } \leq \abs{ \ln\left[ \mu \abs{ u(y,z(\tau)) } \right] } + \frac{\pi}{2} \eqend{,}
\end{equation}
which holds in the limit $\epsilon \to 0$. The remaining dependence on $x$ can then be absorbed in the adiabatic cutoff functions $g(x)$ as in the proof of~\cite[Thm.~5]{froebcadamuro2022}. For the remaining integral over $u(z(\tau))$ or $v(z(\tau))$, we use the bound~\cite[Eq.~(180)]{froebcadamuro2022}
\begin{equation}
\label{eq:mink_conv_num_2_int}
\int \frac{\abs{ \ln\abs{ \mu u(z,x) } }^k}{1 + \mu^2 u(z)^2} \total u(z) \leq \frac{2}{\mu} c_k \ln^k\left( 2 + \mu \abs{u(x)} \right)
\end{equation}
and $\int [ 1 + \mu^2 u(z)^2 ]^{-1} \total u(z) = \pi/\mu$. As in the proof of~\cite[Thm.~5]{froebcadamuro2022}, the logarithm can again be absorbed in the adiabatic cutoff functions $g(x)$.

For the Feynman parametrices $H^\text{F}$, using~\eqref{eq:hf_der_u_tau} instead of~\eqref{eq:hadamard_der_u_tau} the same procedure leads to
\begin{splitequation}
\label{eq:thm_renorm_qed_bound_wh2}
&\int \left. \frac{1}{2} \left( v^\mu v^\nu + w^\mu w^\nu \right) \partial^z_\mu W(x,z) \partial^z_\nu H^\text{F}(y,z) \right\rvert_{z = z(\tau)} f(\tau) \total \tau \\
&= - \frac{\mathi}{4 \pi} \int \Big[ f(\tau) \partial_\tau \Big[ \partial_{u(z)} W(x,z) \Big]_{z = z(\tau)} + f(\tau) \partial_\tau \ln(v^0 - v^1) \Big[ \partial_{u(z)} W(x,z) \Big]_{z = z(\tau)} \\
&\qquad\qquad+ \partial_\tau f(\tau) \Big[ \partial_{u(z)} W(x,z) \Big]_{z = z(\tau)} \Big] \Big[ \Theta(y^0+z^0(\tau)) \ln\left[ \mu ( \epsilon + \mathi u(y,z(\tau)) ) \right] \\
&\qquad\qquad+ \Theta(-y^0-z^0(\tau)) \ln\left[ \mu ( \epsilon - \mathi u(y,z(\tau)) ) \right] \Big] \total u(z(\tau)) \\
&\quad- \frac{\mathi}{4 \pi} \int \Big[ f(\tau) \partial_\tau \Big[ \partial_{v(z)} W(x,z) \Big]_{z = z(\tau)} + f(\tau) \partial_\tau \ln(v^0 + v^1) \Big[ \partial_{v(z)} W(x,z) \Big]_{z = z(\tau)} \\
&\qquad\qquad+ \partial_\tau f(\tau) \Big[ \partial_{v(z)} W(x,z) \Big]_{z = z(\tau)} \Big] \Big[ \Theta(y^0+z^0(\tau)) \ln\left[ \mu ( \epsilon + \mathi v(y,z(\tau)) ) \right] \\
&\qquad\qquad+ \Theta(-y^0-z^0(\tau)) \ln\left[ \mu ( \epsilon - \mathi v(y,z(\tau)) ) \right] \Big] \total v(z(\tau)) \\
&\quad- \frac{1}{4} \int v^0 \delta(y^0-z^0(\tau)) \sgn[ y^1 - z^1(\tau) ] \\
&\qquad\quad\times \Big[ (v^0 + v^1) \partial_{v(z)} W(x,z) - (v^0 - v^1) \partial_{u(z)} W(x,z) \Big]_{z = z(\tau)} f(\tau) \total \tau \eqend{,} \raisetag{1.8em}
\end{splitequation}
and the first two terms can be bounded in the same way as for the Wightman parametrix, using that $\abs{ \Theta(x) } \leq 1$. For the last term, we use the $\delta$ to perform the integral over $\tau$, using~\eqref{eq:dudv_dtau} to change variables. This gives
\begin{equation}
v^0 \total \tau = \total z^0(\tau) \eqend{,}
\end{equation}
and then we can estimate this term by its supremum over all $\tau$:
\begin{splitequation}
&\bigg\lvert - \frac{1}{4} \int v^0 \delta(y^0-z^0(\tau)) \sgn[ y^1 - z^1(\tau) ] \\
&\qquad\times \Big[ (v^0 + v^1) \partial_{v(z)} W(x,z) - (v^0 - v^1) \partial_{u(z)} W(x,z) \Big]_{z = z(\tau)} f(\tau) \total \tau \bigg\rvert \\
&\leq \frac{1}{4} \sup_\tau \abs{ \Big[ (v^0 + v^1) \partial_{v(z)} W(x,z) - (v^0 - v^1) \partial_{u(z)} W(x,z) \Big]_{z = z(\tau)} f(\tau) } \eqend{.}
\end{splitequation}
Again, the dependence on $x$ can be absorbed in the adiabatic cutoff functions $g(x)$ as in the proof of~\cite[Thm.~5]{froebcadamuro2022}.

Lastly, we need to consider terms with products of two Hadamard parametrices appearing in the first term of~\eqref{eq:theta_00_def}, either at the same point or at two different points. We first compute the analogue of~\eqref{eq:thm_renorm_qed_whder}
\begin{splitequation}
\label{eq:thm_renorm_qed_hhder}
&\int \left. \frac{1}{2} \left( v^\mu v^\nu + w^\mu w^\nu \right) \partial^z_\mu H(x,z) \partial^z_\nu H(y,z) \right\rvert_{z = z(\tau)} f(\tau) \total \tau \\
&= \int \Big[ (v^0 - v^1)^2 \partial_{u(z)} H(x,z) \partial_{u(z)} H(y,z) \\
&\qquad+ (v^0 + v^1)^2 \partial_{v(z)} H(x,z) \partial_{v(z)} H(y,z) \Big]_{z = z(\tau)} f(\tau) \total \tau \eqend{,}
\end{splitequation}
and then start with the case of equal points $x = y$. As in the space-time smearing case, only products of two Wightman parametrices or two Feynman parametrices appear, for which we have computed the relevant expressions~\eqref{eq:thm_renorm_hmuhnu_ren} and~\eqref{eq:thm_renorm_hadamard_uu_ren} for products of their derivatives. Again, the logarithmically divergent term appearing for the Feynman parametrices~\eqref{eq:thm_renorm_hmuhnu} drops out since it is proportional to the metric $\eta_{\mu\nu}$ and thus vanishes when contracted with $v^\mu v^\nu + w^\mu w^\nu$, and we also recall that we have used a redefinition of time-ordered products to get rid of the local terms in~\eqref{eq:thm_renorm_hmuhnu_ren}. It follows that both for the Wightman and Feynman parametrices, we have
\begin{splitequation}
&\int \left. \frac{1}{2} \left( v^\mu v^\nu + w^\mu w^\nu \right) \left[ \partial^z_\mu H(x,z) \partial^z_\nu H(x,z) \right]^\text{ren} \right\rvert_{z = z(\tau)} f(\tau) \total \tau \\
&= \int \left. \frac{1}{2} \left( v^\mu v^\nu + w^\mu w^\nu \right) \partial^z_\mu H(x,z) \partial^z_\nu H(x,z) \right\rvert_{z = z(\tau)} f(\tau) \total \tau \\
&= - \frac{\mathi}{4 \pi} \int \Big[ (v^0 - v^1)^2 \partial_{u(z)}^2 H(x,z) + (v^0 + v^1)^2 \partial_{v(z)}^2 H(x,z) \Big]_{z = z(\tau)} f(\tau) \total \tau \eqend{,}
\end{splitequation}
and analogously to~\eqref{eq:hadamard_der_u_tau} and~\eqref{eq:hf_der_u_tau} we convert the $u$ and $v$ derivatives into $\tau$ derivatives. We compute
\begin{splitequation}
&\left[ (v^0-v^1)^2 \partial_u^2 H(x,z) + (v^0+v^1)^2 \partial_v^2 H(x,z) \right]_{z = z(\tau)} \\
&\quad= \partial_\tau^2 H(x,z(\tau)) - 2 \left[ \partial_{u(z)} \partial_{v(z)} H(x,z) \right]_{z = z(\tau)} \\
&\qquad- [ \partial_\tau (v^0-v^1) ] \left[ \partial_{u(z)} H(x,z) \right]_{z = z(\tau)} - [ \partial_\tau (v^0+v^1) ] \left[ \partial_{v(z)} H(x,z) \right]_{z = z(\tau)} \eqend{,}
\end{splitequation}
and for a Wightman parametrix we have $\partial_{u(z)} \partial_{v(z)} H^+(x,z) = 0$, while for a Feynman parametrix $\partial_{u(z)} \partial_{v(z)} H^\text{F}(x,z) = - \frac{1}{4} \delta^2(x-z)$~\eqref{eq:hf_fundamental_solution}. The remaining $u$ and $v$ derivatives we can convert into $\tau$ derivatives using~\eqref{eq:hadamard_der_u_tau} and~\eqref{eq:hf_der_u_tau}, partially integrate all $\tau$ derivatives to act on the test function $f$ or the $v^\mu(\tau)$, and then bound the result completely analogous to the previous case of a product of $W$ and one parametrix.

The only potentially problematic term is the local one $\sim \delta^2(x-z(\tau))$ that appears for a product of Feynman parametrices, whose contribution to~\eqref{eq:thm_renorm_qed_hhder} (for the Feynman parametrix) is
\begin{equation}
- \frac{\mathi}{8 \pi} \int \delta^2(x-z(\tau)) f(\tau) \total \tau \eqend{,}
\end{equation}
and whose contribution to the first term of~\eqref{eq:theta_00_def} and thus the expectation value of the interacting energy density~\eqref{eq:eomega_series} reads
\begin{splitequation}
\label{eq:thm_renorm_qed_dhdh_local}
&- \frac{\mathi}{8 \pi} \beta^2 \sum_{n=0}^\infty \sum_{k=0}^{2n} \sum_{\ell=\max(0,k-n)}^{\min(k,n)} \frac{(-1)^{k+n}}{\ell! (k-\ell)! (n-\ell)! (n-k+\ell)!} \int\dotsi\int \\
&\qquad\times \int \left[ \sum_{i=\ell+1}^n \delta^2(x_i-z(\tau)) + \sum_{j=k-\ell+1}^n \delta^2(y_j-z(\tau)) \right] f(\tau) \total \tau \\
&\qquad\times \mathcal{E}_{k,\ell,2n,n-\ell}(\vec{x};\vec{y}) \prod_{i=1}^n g(x_i) g(y_i) \total^2 x_i \total^2 y_i \eqend{,}
\end{splitequation}
where because of the Kronecker $\delta_{2(\ell+m),n}$ we replaced $n \to 2n$ and then performed the sum over $m$. Let us thus consider one of the terms, say with $\delta^2(x_n-z(\tau))$ which allows us to perform the integral over $x_n$ and bound the result by
\begin{splitequation}
&\frac{\beta^2}{8 \pi} \sum_{n=0}^\infty \sum_{k=0}^{2n} \sum_{\ell=\max(0,k-n)}^{\min(k,n)} \frac{1}{\ell! (k-\ell)! (n-\ell)! (n-k+\ell)!} \int\dotsi\int \abs{ \mathcal{E}_{k,\ell,2n,n-\ell}(\hat{\vec{x}};\vec{y}) } \\
&\quad\times \prod_{i=1}^{n-1} \abs{ g(x_i) } \abs{ g(y_i) } \total^2 x_i \total^2 y_i \, \abs{ g(z(\tau)) } \abs{ f(\tau) } \total \tau \, \abs{ g(y_n) } \total^2 y_n \eqend{,} \raisetag{2.4em}
\end{splitequation}
where $\hat{\vec{x}} = \{x_1,\ldots,x_{n-1},z(\tau)\}$. However, after the shift $n \to n+1$ this is now of the same form as the bounds~\eqref{eq:thm_renorm_vertex_bound} for the vertex operators in~\eqref{eq:theta_00_def}, which are given by~\eqref{eq:thm_renorm_vertex_bound2}, and so also the local terms are suitably bounded, taking into account that there $n-\ell + n-(k-\ell) = 2n-k \leq 2n$ of them~\eqref{eq:thm_renorm_qed_dhdh_local}.

Finally, we need to consider terms with products of two Hadamard parametrices at two different points, for which we still use~\eqref{eq:thm_renorm_qed_hhder}. However, now we use that~\cite[Eq.~(198)]{froebcadamuro2022}
\begin{equation}
\partial_{u(z)} H^\text{F}(z,x) \partial_{u(z)} H^\text{F}(z,y) = \frac{1}{32 \pi^2} \partial_{u(z)}^2 \ln\left[ \mu^2 u(z,a)^2 \right]
\end{equation}
for some point $a$ such that $\min(u(x),u(y)) \leq u(a) \leq \max(u(x),u(y))$, and the analogous result for $v$ or for the product of two Wightman parametrices. We can then convert the $u(z)$ and $v(z)$ derivatives into $\tau$ derivatives using~\eqref{eq:dudv_dtau}, and obtain
\begin{splitequation}
&\int \left. \frac{1}{2} \left( v^\mu v^\nu + w^\mu w^\nu \right) \partial^z_\mu H(x,z) \partial^z_\nu H(y,z) \right\rvert_{z = z(\tau)} f(\tau) \total \tau \\
&= \frac{1}{32 \pi^2} \int \Big[ \partial_\tau^2 \ln\left[ \mu^2 u(z(\tau),a)^2 \right] - \partial_\tau \ln(v^0-v^1) \partial_\tau \ln\left[ \mu^2 u(z(\tau),a)^2 \right] \\
&\qquad+ \partial_\tau^2 \ln\left[ \mu^2 v(z(\tau),a')^2 \right] - \partial_\tau \ln(v^0+v^1) \partial_\tau \ln\left[ \mu^2 v(z(\tau),a')^2 \right] \Big] f(\tau) \total \tau \eqend{,}
\end{splitequation}
where $a'$ may be different from $a$. Integrating the $\tau$ derivatives by parts and changing integration variables to either $u(z(\tau))$ or $v(z(\tau))$ depending on the argument of the logarithm [using~\eqref{eq:dudv_dtau}], this gives
\begin{splitequation}
&\int \left. \frac{1}{2} \left( v^\mu v^\nu + w^\mu w^\nu \right) \partial^z_\mu H(x,z) \partial^z_\nu H(y,z) \right\rvert_{z = z(\tau)} f(\tau) \total \tau \\
&= \frac{1}{32 \pi^2} \int \frac{1}{v^0 - v^1} \Big[ f''(\tau) + \partial_\tau \left[ f(\tau) \partial_\tau \ln(v^0-v^1) \right] \Big] \ln\left[ \mu^2 u(z(\tau),a)^2 \right] \total u(z(\tau)) \\
&\quad+ \frac{1}{32 \pi^2} \int \frac{1}{v^0 + v^1} \Big[ f''(\tau) + \partial_\tau \left[ f(\tau) \partial_\tau \ln(v^0+v^1) \right] \Big] \ln\left[ \mu^2 v(z(\tau),a')^2 \right] \total v(z(\tau)) \eqend{,}
\end{splitequation}
where now $\tau$ is a function of $u(z(\tau))$ or $v(z(\tau))$. We then estimate the terms apart from the logarithm by their supremum:
\begin{splitequation}
\label{eq:mink_conv_ddf_bound}
&\abs{ \frac{1}{v^0 - v^1} \Big[ f''(\tau) + \partial_\tau \left[ f(\tau) \partial_\tau \ln(v^0-v^1) \right] \Big] } \\
&\leq \frac{1}{1 + \mu^2 u(z(\tau))^2} \sup_\tau \abs{ \frac{1 + \mu^2 u(z(\tau))^2}{v^0 - v^1} \Big[ f''(\tau) + \partial_\tau \left[ f(\tau) \partial_\tau \ln(v^0-v^1) \right] \Big] } \eqend{,}
\end{splitequation}
and then use~\eqref{eq:mink_conv_num_2_int} to perform the integral over $u(z(\tau))$, and the analogous estimate for the integral over $v(z(\tau))$. Finally, using that $\min(u(x),u(y)) \leq u(a) \leq \max(u(x),u(y))$ we estimate the resulting logarithm by
\begin{equation}
\ln\left( 2 + \mu \abs{u(a)} \right) \leq \ln\left( 2 + \mu \abs{u(x)} \right) + \ln\left( 2 + \mu \abs{u(y)} \right) \eqend{,}
\end{equation}
and the logarithms can then be absorbed into the adiabatic cutoff functions $g(x)$. We see at this point that the bound~\eqref{eq:mink_conv_ddf_bound} diverges in the limit where the trajectory becomes light-like, $v^1 \to \pm v^0$. This precludes the direct extension of our result to the null case, but for arbitrary time-like trajectories $z(\tau)$ that we consider in this work we of course have a finite bound.

We then follow the same steps as in the proof of~\cite[Thm.~5]{froebcadamuro2022} and use that there are in total $4 n^2$ products of Feynman or Wightman parametrices and $W$ in the product of derivatives of $\mathcal{G}_{k,\ell,n,m}$~\eqref{eq:calg_def} in $\Theta^{k,\ell,n,m}$~\eqref{eq:theta_00_def}. We can thus bound the contribution of these terms to the expectation value of the energy density at order $n$ by $C n^2 K^n (n!)^{1 + \frac{\beta^2}{4\pi}}$ as required, where the constants $C$ and $K$ depend on $f$, $g$, $\beta$, $W$ and $z(\tau)$. We note however that the dependence on $W$ only comes from the explicit $W$ contained in $\mathcal{G}_{k,\ell,n,m}$, since due to our condition on $W$~\eqref{eq:thm_renorm_wpos_disc} the $W$ dependence in $\mathcal{E}_{k,\ell,n,m}$ is bounded by $1$ analogously to the bounds~\eqref{eq:thm_renorm_vertex_bound2} for the contribution from the vertex operators. Taking all together, we obtain a bound of the same form as for a space-time smearing~\eqref{eq:thm_renorm_bound}, and thus following the same steps to estimate the sums we obtain the bound
\begin{equation}
\abs{ E_\omega(f) } \leq C' \sum_{n=0}^\infty \frac{(n+1)^2 (K')^{2n}}{(n!)^{1-\frac{\beta^2}{4 \pi}}} < \infty
\end{equation}
on the smeared quantum energy density, where $C'$ and $K'$ are new constants depending on $f$, $g$, $\beta$, $W$ and $z(\tau)$. \hfill\squareforqed

\section{Proof of Theorem~\ref{thm_qei} (Quantum energy inequality)}
\label{sec_qei}

To obtain the quantum energy inequality~\eqref{eq:thm_qei} in the interacting theory, we suitably extend methods known from the free theory~\cite{flanagan1997,fewster2000,fewstersmith2008}, namely writing the energy density in point-split form and then separating it into a sum of a positive and a bounded part, with the bounded part being state-independent.

We first note that the energy density when smeared with the square of a test function $f \in \mathcal{S}(\mathbb{R})$ can be written in a point-split form
\begin{equation}
E_\omega(f^2) = \lim_{\Lambda,\epsilon \to 0} \omega^{\Lambda,\epsilon}\left( \mathcal{T}_\text{int}\left[ \hat{T}_{\mu\nu}( v^\mu v^\nu f^2 ) \right] \right) = \int e_{\omega,f}(\tau,\tau) \total \tau \eqend{,}
\end{equation}
where
\begin{equation}
\label{eq:qei_pointsplit_e_def}
e_{\omega,f}(\tau,\tau') \equiv v^\mu(\tau) v^\nu(\tau') f(\tau) f(\tau') \lim_{\Lambda,\epsilon \to 0} \omega^{\Lambda,\epsilon}\left( \mathcal{T}_\text{int}\left[ \hat{T}^\parallel_{\mu\nu}( z(\tau), z(\tau') ) \right] \right)
\end{equation}
with the point-split stress tensor
\begin{splitequation}
\label{eq:qei_pointsplit_tmunu_pre}
\hat{T}^\parallel_{\mu\nu}(z,z') &\equiv \partial_\mu \phi(z) \partial_\nu \phi(z') - \frac{1}{2} \eta_{\mu\nu} \partial^\rho \phi(z) \partial_\rho \phi(z') \\
&\quad+ \frac{1}{2} \left( 1 - \frac{\beta^2}{8 \pi} \right) \eta_{\mu\nu} \Big[ g(z) ( V_\beta(z) + V_{-\beta}(z) ) + g(z') ( V_\beta(z') + V_{-\beta}(z') ) \Big] \eqend{.}
\end{splitequation}
Actually, since the time-ordered products involve renormalisation, and in particular the subtraction of the Hadamard parametrix which is singular as the two points approach each other, it is more appropriate to define the point-split value of the interacting time-ordered product of the point-split stress tensor instead of the point-split stress tensor itself. So instead of Lemma~\ref{lemma1}, in particular the interacting time-ordered product $\mathcal{T}_\text{int}\left[ \op_{\mu\nu}(z) \right]$~\eqref{eq:tint_2phi_def}, we define a point-split time-ordered product by
\begin{splitequation}
\label{eq:tint_2phi_pointsplit_def}
&\mathcal{T}_\text{int}\left[ \op^\parallel_{\mu\nu}(z,z') \right] \equiv \sum_{n=0}^\infty \sum_{k=0}^n \sum_{\ell=0}^k \sum_{m=0}^{n-k} \frac{(-1)^k \mathi^n}{\ell! (k-\ell)! m! (n-k-m)!} \int\dotsi\int \mathcal{E}_{k,\ell,n,m}(\vec{x};\vec{y}) \\
&\quad\times \Bigg[ \beta^2 \partial^z_\mu \mathcal{G}_{k,\ell,n,m}(\vec{x};\vec{y};z) \partial^{z'}_\nu \mathcal{G}_{k,\ell,n,m}(\vec{x};\vec{y};z') \, \mathcal{N}_G\left[ \prod_{i=1}^{\ell+m} V_\beta(x_i) \prod_{j=1}^{n-\ell-m} V_{-\beta}(y_j) \right] \\
&\qquad\quad- \beta \partial^z_\mu \mathcal{G}_{k,\ell,n,m}(\vec{x};\vec{y};z) \, \mathcal{N}_G\left[ \partial_\nu \phi(z') \prod_{i=1}^{\ell+m} V_\beta(x_i) \prod_{j=1}^{n-\ell-m} V_{-\beta}(y_j) \right] \\
&\qquad\quad- \beta \partial^{z'}_\nu \mathcal{G}_{k,\ell,n,m}(\vec{x};\vec{y};z') \, \mathcal{N}_G\left[ \partial_\mu \phi(z) \prod_{i=1}^{\ell+m} V_\beta(x_i) \prod_{j=1}^{n-\ell-m} V_{-\beta}(y_j) \right] \\
&\qquad\quad+ \mathcal{N}_G\left[ \partial_\mu \phi(z) \partial_\nu \phi(z') \prod_{i=1}^{\ell+m} V_\beta(x_i) \prod_{j=1}^{n-\ell-m} V_{-\beta}(y_j) \right] \\
&\qquad\quad+ \partial^z_\mu \partial^{z'}_\nu W(z,z') \, \mathcal{N}_G\left[ \prod_{i=1}^{\ell+m} V_\beta(x_i) \prod_{j=1}^{n-\ell-m} V_{-\beta}(y_j) \right] \Bigg] \\
&\quad\times \left( \frac{\Lambda}{\mu} \right)^\frac{\beta^2 (2(\ell+m)-n)^2}{4 \pi} \prod_{i=1}^{\ell+m} g(x_i) \total^2 x_i \prod_{j=1}^{n-\ell-m} g(y_j) \total^2 y_j \eqend{,} \raisetag{2.4em}
\end{splitequation}
and the point-split interacting time-ordered product of the stress tensor by
\begin{splitequation}
\label{eq:qei_pointsplit_tmunu}
&\mathcal{T}_\text{int}\left[ \hat{T}^\parallel_{\mu\nu}(z,z') \right] \equiv \mathcal{T}_\text{int}\left[ \op^\parallel_{\mu\nu}(z,z') \right] - \frac{1}{2} \eta_{\mu\nu} \eta^{\rho\sigma} \mathcal{T}_\text{int}\left[ \op^\parallel_{\rho\sigma}(z,z') \right] \\
&\qquad+ \frac{1}{2} \left( 1 - \frac{\beta^2}{8 \pi} \right) \eta_{\mu\nu} \Big[ g(z) ( \mathcal{T}_\text{int}\left[ V_\beta(z) \right] + \mathcal{T}_\text{int}\left[ V_{-\beta}(z) \right] ) + (z \leftrightarrow z') \Big] \eqend{.}
\end{splitequation}
Since the normal-ordered products of operators are continuous functions of the insertion points and $W$ is smooth, almost all terms in the point-split time-ordered product~\eqref{eq:tint_2phi_pointsplit_def} have a finite limit as $z' \to z$. Only the products of derivatives of $\mathcal{G}_{k,\ell,m,n}$~\eqref{eq:calg_def} which contain Hadamard parametrices need renormalisation in the limit where the UV cutoff $\epsilon \to 0$, but as we have seen in the proof of Theorem~\ref{thm_renorm}, the required counterterms are proportional to $\eta_{\mu\nu}$~\eqref{eq:thm_renorm_hmuhnu} and thus cancel in the stress tensor. We thus have
\begin{equation}
\lim_{z' \to z} \mathcal{T}_\text{int}\left[ \hat{T}^\parallel_{\mu\nu}(z,z') \right] = \mathcal{T}_\text{int}\left[ \hat{T}_{\mu\nu}(z) \right]
\end{equation}
with the point-split stress tensor~\eqref{eq:qei_pointsplit_tmunu}, and moreover the point-split quantity $e_{\omega,f}(\tau,\tau')$~\eqref{eq:qei_pointsplit_e_def} is continuous in $\tau$ and $\tau'$. All the hypotheses of~\cite[Lemma~2.12]{fewstersmith2008} are thus fulfilled and we can apply it and obtain
\begin{equation}
\label{eq:qei_e_fourier}
E_\omega(f^2) = \int e_{\omega,f}(\tau,\tau) \total \tau = \lim_{\delta \to 0^+} \int \mathe^{- \delta \xi^2} (\mathcal{F} e_{\omega,f})(-\xi,\xi) \frac{\total \xi}{2 \pi} \eqend{,}
\end{equation}
where the Fourier transform is defined by
\begin{equation}
\label{eq:qei_fourier_def}
(\mathcal{F} g)(\xi,\zeta) \equiv \iint \mathe^{\mathi \xi \tau + \mathi \zeta \sigma} g(\tau,\sigma) \total \tau \total \sigma \eqend{.}
\end{equation}
Since~\cite[Lemma~2.12]{fewstersmith2008} is only applicable to functions of compact support, we temporarily restrict to test functions $f$ of compact support, and obtain the final result for arbitrary Schwartz functions $f$ by a limit argument. Moreover, since $e_{\omega,f}(\tau,\tau')$ is by construction symmetric in $\tau$ and $\tau'$, the integrand in~\eqref{eq:qei_e_fourier} is invariant under the change $\xi \to -\xi$, and so we may restrict the integration to positive $\xi$, obtaining a factor of 2. That is, to prove the QEI~\eqref{eq:thm_qei} we need to show that
\begin{equation}
\label{eq:thm_qei_N}
E^N_\omega(f^2) = 2 \lim_{\delta \to 0^+} \int_0^\infty \mathe^{- \delta \xi^2} (\mathcal{F} e_{\omega,f_N})(-\xi,\xi) \frac{\total \xi}{2 \pi} \geq - K(z,f_N)
\end{equation}
and that $\lim_{N \to \infty} K(z,f_N) < \infty$, where $f_N \in \mathcal{S}_\text{c}(\mathbb{R})$ is a sequence of test functions of compact support converging to $f$.\footnote{We could also restrict to negative $\xi$, but this would ultimately result in a trivial QEI with $\lim_{N \to \infty} K(z,f_N) = \infty$~\cite{fewstersmith2008}.}

However, we will need~\cite[Lemma~2.12]{fewstersmith2008} only for some part of the bound, where it is technically difficult to take the limit $\tau' \to \tau$ directly in position space and then integrate over $\tau$. To show the quantum energy inequality~\eqref{eq:thm_qei}, we first split the energy density $e_{\omega,f}$~\eqref{eq:qei_pointsplit_e_def} into a positive part and a part that we can bound independently of the state. The latter part is furthermore split into several terms, one of which requires the use of~\cite[Lemma~2.12]{fewstersmith2008} in the form~\eqref{eq:thm_qei_N}, while the other ones can be bounded directly. We will see that the positive part is given by
\begin{equation}
\frac{1}{2} \Big[ v^\mu(\tau) v^\nu(\tau') + w^\mu(\tau) w^\nu(\tau') \Big] \, \Big( \mathcal{T}_\text{int}\left[ \partial_\mu \phi(z(\tau)) \right] \Big)^\dagger \star \mathcal{T}_\text{int}\left[ \partial_\nu \phi(z(\tau')) \right] \eqend{,}
\end{equation}
the star product of two interacting time-ordered products of derivatives of the field $\partial_\mu \phi$, for which we could use the result~\eqref{eq:tint_1phi_def}. However, the computation is shortened by inserting the Bogoliubov formula~\eqref{eq:tint_def} and using that
\begin{equation}
\Big( \overline{\mathcal{T}}\left[ \op_1(f_1) \otimes \cdots \otimes \op_k(f_k) \right] \Big)^\dagger = \mathcal{T}\left[ \op_1^\dagger(f_1^*) \otimes \cdots \otimes \op_k^\dagger(f_k^*) \right]
\end{equation}
as well as the star inverse~\eqref{eq:t_starinverse}, such that
\begin{equation}
\Big( \mathcal{T}_\text{int}\left[ \partial_\mu \phi(z) \right] \Big)^\dagger \star \mathcal{T}_\text{int}\left[ \partial_\nu \phi(z') \right] = \overline{\mathcal{T}}\left[ \mathe_\otimes^{ - \mathi S_\text{int} } \otimes \partial_\mu \phi(z) \right] \star \mathcal{T}\left[ \mathe_\otimes^{ \mathi S_\text{int} } \otimes \partial_\nu \phi(z') \right] \eqend{.}
\end{equation}
A quite long but straightforward computation using the result~\eqref{eq:timeordered_exponential_1phi}, the analogously determined
\begin{splitequation}
\label{eq:antitimeordered_exponential_1phi}
&\overline{\mathcal{T}}\left[ \bigotimes_{j=1}^n V_{\alpha_j}(x_j) \otimes \phi(z) \right] = \exp\left[ - \mathi \sum_{1 \leq i < j \leq n} \alpha_i \alpha_j H^\text{D}(x_i,x_j) \right] \\
&\qquad\times\Bigg[ \mathcal{N}_G\left[ \phi(z) \prod_{j=1}^n V_{\alpha_j}(x_j) \right] - \sum_{i=1}^n \alpha_i G^\text{D}(x_i,z) \mathcal{N}_G\left[ \prod_{j=1}^n V_{\alpha_j}(x_j) \right] \Bigg] \\
&\qquad\times \exp\left[ - \frac{1}{2} \sum_{i,j=1}^n \alpha_i \alpha_j W(x_i,x_j) \right] \left( \frac{\Lambda}{\mu} \right)^\frac{\left( \sum_{k=1}^n \alpha_k \right)^2}{4 \pi} \eqend{,}
\end{splitequation}
the formula~\eqref{eq:normal_ordering_starproduct} for the star product and the explicit form of the two-point function~\eqref{eq:twopf_hadamard} then gives
\begin{splitequation}
\label{eq:qei_tint_dphi2}
&\Big( \mathcal{T}_\text{int}\left[ \partial_\mu \phi(z) \right] \Big)^\dagger \star \mathcal{T}_\text{int}\left[ \partial_\nu \phi(z') \right] \\
&= \sum_{n=0}^\infty \sum_{k=0}^n \sum_{\ell=0}^k \sum_{m=0}^{n-k} \frac{(-1)^k \mathi^n}{\ell! (k-\ell)! m! (n-k-m)!} \int\dotsi\int \mathcal{E}_{k,\ell,n,m}(\vec{x};\vec{y}) \\
&\quad\times\Bigg[ \beta^2 \partial^z_\mu \overline{\mathcal{G}}_{k,\ell,n,m}(\vec{x};\vec{y};z) \partial^{z'}_\nu \mathcal{G}_{k,\ell,n,m}(\vec{x};\vec{y};z') \, \mathcal{N}_G\left[ \prod_{i=1}^{\ell+m} V_\beta(x_i) \prod_{j=1}^{n-\ell-m} V_{-\beta}(y_j) \right] \\
&\qquad\quad- \beta \partial^z_\mu \overline{\mathcal{G}}_{k,\ell,n,m}(\vec{x};\vec{y};z) \, \mathcal{N}_G\left[ \partial_\nu \phi(z') \prod_{i=1}^{\ell+m} V_\beta(x_i) \prod_{j=1}^{n-\ell-m} V_{-\beta}(y_j) \right] \\
&\qquad\quad- \beta \partial^{z'}_\nu \mathcal{G}_{k,\ell,n,m}(\vec{x};\vec{y};z') \, \mathcal{N}_G\left[ \partial_\mu \phi(z) \prod_{i=1}^{\ell+m} V_\beta(x_i) \prod_{j=1}^{n-\ell-m} V_{-\beta}(y_j) \right] \\
&\qquad\quad+ \mathcal{N}_G\left[ \partial_\mu \phi(z) \partial_\nu \phi(z') \prod_{i=1}^{\ell+m} V_\beta(x_i) \prod_{j=1}^{n-\ell-m} V_{-\beta}(y_j) \right] \\
&\qquad\quad+ \mathi \partial^z_\mu \partial^{z'}_\nu G^+(z,z') \, \mathcal{N}_G\left[ \prod_{i=1}^{\ell+m} V_\beta(x_i) \prod_{j=1}^{n-\ell-m} V_{-\beta}(y_j) \right] \Bigg] \\
&\quad\times \left( \frac{\Lambda}{\mu} \right)^\frac{\beta^2 (2(\ell+m)-n)^2}{4 \pi} \prod_{i=1}^{\ell+m} g(x_i) \total^2 x_i \prod_{j=1}^{n-\ell-m} g(y_j) \total^2 y_j \eqend{,} \raisetag{2.3em}
\end{splitequation}
where we used that the (anti-)time-ordered products are symmetric in their arguments, and that there are $k!/[ \ell! (k-\ell)! ]$ ways to choose $\ell$ vertex operators $V_\beta$ with positive sign from a total of $k$ ones $V_{\pm \beta}$ with either sign. We also renamed the insertion points of the $V_{-\beta}$ operators to $y_j$, and recall that $\mathcal{E}_{k,\ell,n,m}$ and $\mathcal{G}_{k,\ell,n,m}$ are defined by~\eqref{eq:cale_def} and~\eqref{eq:calg_def}, while
\begin{splitequation}
\label{eq:barcalg_def}
\overline{\mathcal{G}}_{k,\ell,n,m}(\vec{x};\vec{y};z) &\equiv \sum_{i=1}^\ell H^\text{D}(z,x_i) + \sum_{i=\ell+1}^{\ell+m} H^+(z,x_i) - \sum_{j=1}^{k-\ell} H^\text{D}(z,y_j) \\
&\quad- \sum_{j=k-\ell+1}^{n-\ell-m} H^+(z,y_j) - \mathi \sum_{i=1}^{\ell+m} W(z,x_i) + \mathi \sum_{j=1}^{n-\ell-m} W(z,y_j) \eqend{.}
\end{splitequation}
Contracting with $v^\mu v^\nu$ and comparing with~\eqref{eq:tint_2phi_pointsplit_def}, we see that the point-split quantum energy density~\eqref{eq:qei_pointsplit_e_def} with the interacting time-ordered product of the stress tensor~\eqref{eq:qei_pointsplit_tmunu} can be written as
\begin{splitequation}
\label{eq:qei_e_positivebounded}
e_{\omega,f}(\tau,\tau') &= \frac{1}{2} f(\tau) f(\tau') \Big[ v^\mu(\tau) v^\nu(\tau') + w^\mu(\tau) w^\nu(\tau') \Big] \\
&\qquad\times \lim_{\Lambda,\epsilon \to 0} \omega^{\Lambda,\epsilon}\left( \Big( \mathcal{T}_\text{int}\left[ \partial_\mu \phi(z) \right] \Big)^\dagger \star \mathcal{T}_\text{int}\left[ \partial_\nu \phi(z') \right] \right) \\
&\quad+ \frac{1}{2} f(\tau) f(\tau') \Big[ v^\mu(\tau) v^\nu(\tau') + w^\mu(\tau) w^\nu(\tau') \Big] \\
&\qquad\times \sum_{n=0}^\infty \sum_{k=0}^n \sum_{\ell=0}^k \sum_{m=0}^{n-k} \frac{(-1)^k \mathi^n \delta_{2(\ell+m),n}}{\ell! (k-\ell)! m! (n-k-m)!} \int\dotsi\int \mathcal{E}_{k,\ell,n,m}(\vec{x};\vec{y}) \\
&\quad\times \bigg[ \beta^2 \partial^z_\mu \left[ \mathcal{G}_{k,\ell,n,m}(\vec{x};\vec{y};z) - \overline{\mathcal{G}}_{k,\ell,n,m}(\vec{x};\vec{y};z) \right] \partial^{z'}_\nu \mathcal{G}_{k,\ell,n,m}(\vec{x};\vec{y};z') \\
&\qquad\quad- \mathi \partial^z_\mu \partial^{z'}_\nu H^+(z,z') \bigg] \prod_{i=1}^{\ell+m} g(x_i) \total^2 x_i \prod_{j=1}^{n-\ell-m} g(y_j) \total^2 y_j \\
&\quad- \frac{1}{2} \left( 1 - \frac{\beta^2}{8 \pi} \right) f(\tau) f(\tau') \Big[ g(z(\tau)) \lim_{\Lambda,\epsilon \to 0} \\
&\qquad\times \omega^{\Lambda,\epsilon}\left( \mathcal{T}_\text{int}\left[ V_\beta(z(\tau)) \right] + \mathcal{T}_\text{int}\left[ V_{-\beta}(z(\tau)) \right] \right) + (\tau \leftrightarrow \tau') \Big] \eqend{,} \raisetag{1.6em}
\end{splitequation}
where after taking the derivatives we set $z = z(\tau)$ and $z' = z(\tau')$. Here, we used again that only neutral configurations remain in the limit $\Lambda \to 0$. The first term is the positive part, which would be obvious from the positivity of the state $\omega$, except for the IR divergence of the state. However, we will show below that it is indeed positive. The second and third term comprise the bounded part, and we will show below a state-independent bound for it.

Let us first consider the positivity of the first term of~\eqref{eq:qei_e_positivebounded}. Since it is of the form $\omega^{\Lambda,\epsilon}(A^\dagger \star A)$ and $\omega$ is a quasi-free state, it is positive if the two-point function~\eqref{eq:twopf_hadamard} of $\omega$ is positive definite. However, $\omega$ is indefinite since positivity of the two-point function~\eqref{eq:twopf_hadamard} holds only for test functions with vanishing mean~\eqref{eq:thm_renorm_pos_hplusw}: for the Hadamard parametrix $H^+$ by direct computation, and for the state-dependent part $W$ by our condition of conditional positive semidefiniteness. Therefore, to show positivity of~\eqref{eq:qei_e_positivebounded} we need to construct an auxiliary state $\tilde{\omega}$, whose two-point function has the same behaviour as $\omega$ in the limit $\Lambda \to 0$ (i.e., the difference of their two-point functions converges to $0$), but is positive for all $\Lambda > 0$. In principle, this is the massive extension of our massless state, but its construction is not fully trivial. We first use that as $\Lambda \to 0$ we have~\cite[Eq.~(10.31.2)]{dlmf}
\begin{equation}
\bessel*{K}{0}{2 \Lambda \mathe^{-\gamma} \sqrt{ ( \epsilon + \mathi u ) ( \epsilon + \mathi v ) } } = - \frac{1}{2} \ln \left[ \Lambda^2 ( \epsilon + \mathi u ) ( \epsilon + \mathi v ) \right] + \bigo{\Lambda^2 \ln \Lambda}
\end{equation}
for the second modified Bessel function $\mathrm{K}$, and using the integral~\cite[Eq.~(10.32.10)]{dlmf} we compute that
\begin{equation}
\bessel*{K}{0}{M \sqrt{ ( \epsilon + \mathi u ) ( \epsilon + \mathi v ) } } = \int_0^\infty \exp\left[ - \frac{a}{2} M ( \epsilon + \mathi u ) - \frac{M}{2 a} ( \epsilon + \mathi v ) \right] \frac{1}{2 a} \total a
\end{equation}
for $M > 0$. We further compute
\begin{splitequation}
\partial^2 \bessel*{K}{0}{M \sqrt{ ( \epsilon + \mathi u ) ( \epsilon + \mathi v ) } } &= - 4 \partial_u \partial_v \bessel*{K}{0}{M \sqrt{ ( \epsilon + \mathi u ) ( \epsilon + \mathi v ) } } \\
&= M^2 \int_0^\infty \exp\left[ - \frac{a}{2} M ( \epsilon + \mathi u ) - \frac{M}{2 a} ( \epsilon + \mathi v ) \right] \frac{1}{2 a} \total a \\
&= M^2 \bessel*{K}{0}{M \sqrt{ ( \epsilon + \mathi u ) ( \epsilon + \mathi v ) } } \eqend{,} \raisetag{1.5em}
\end{splitequation}
using that the $a$ integral is absolutely convergent for $\epsilon > 0$ such that we can differentiate under the integral sign. It follows that
\begin{equation}
\left( \partial^2 - 4 \Lambda^2 \mathe^{-2\gamma} \right) \bessel*{K}{0}{2 \Lambda \mathe^{-\gamma} \sqrt{ ( \epsilon + \mathi u ) ( \epsilon + \mathi v ) } } = 0 \eqend{,}
\end{equation}
and we see that we have obtained a suitable massive extension of our massless Hadamard parametrix $H^+(x,y)$.

For the state-dependent part $W(x,y)$, the construction is more involved. For this, we use that in two dimensions solutions of the massless Klein--Gordon equation split into functions that depend on only one light-cone coordinate: the general solution of $\partial^2 f(x) = 0$ is $f(x) = f_u(u(x)) + f_v(v(x))$ since the d'Alembertian factorises according to $\partial^2 = - 4 \partial_u \partial_v$. Since $\partial_x^2 W(x,y) = \partial_y^2 W(x,y) = 0$, this generalises to
\begin{splitequation}
W(x,y) &= W_{uu}(u(x),u(y)) + W_{uv}(u(x),v(y)) \\
&\quad+ W_{vu}(v(x),u(y)) + W_{vv}(v(x),v(y)) \eqend{,}
\end{splitequation}
where $W_{uu}$, $W_{uv}$, $W_{vu}$ and $W_{vv}$ are smooth functions. Because of the symmetry $W(x,y) = W(y,x)$, $W_{uu}$ and $W_{vv}$ are symmetric in their arguments, while $W_{vu}(u,v) = W_{uv}(v,u)$. Passing to Fourier space in each variable, we obtain
\begin{equation}
W_{uv}(u(x),v(y)) = \int \tilde{W}_{uv}(p,q) \, \mathe^{\mathi p u(x) + \mathi q v(y)} \frac{\total p \total q}{(2\pi)^2}
\end{equation}
and the analogous expressions for $W_{uu}$, $W_{vu}$ and $W_{vv}$, where the symmetry leads to $\tilde{W}_{vu}(p,q) = \tilde{W}_{uv}(q,p)$. The conditionally positive semidefiniteness~\eqref{eq:thm_renorm_wpos_cont} of $W$ reads in Fourier space
\begin{splitequation}
\label{eq:qei_w_pos}
&\iint \Big[ \tilde{W}_{uu}(p,q) \, \tilde{f}(-p,0) \tilde{f}^*(q,0) + \tilde{W}_{vu}(p,q) \, \tilde{f}(0,-p) \tilde{f}^*(q,0) \\
&\qquad+ \tilde{W}_{uv}(p,q) \, \tilde{f}(-p,0) \tilde{f}^*(0,q) + \tilde{W}_{vv}(p,q) \, \tilde{f}(0,-p) \tilde{f}^*(0,q) \Big] \frac{\total p \total q}{(2\pi)^2} \geq 0 \raisetag{4.5em}
\end{splitequation}
for all $\tilde{f} \istest$ with $\tilde{f}(0,0) = 0$, where
\begin{equation}
\tilde{f}(p,q) \equiv \iint f(u,v) \mathe^{- \mathi p u - \mathi q v} \total u \total v
\end{equation}
is the Fourier transform of the test function $f$. Since $f$ is arbitrary, the condition~\eqref{eq:qei_w_pos} actually separates in individual conditions for $W_{uu}$, $W_{vv}$ and $W_{uv}$:
\begin{equations}[eq:qei_wuuvv_pos]
\iint \tilde{W}_{uu}(p,q) \, \tilde{g}(-p) \tilde{g}^*(q) \frac{\total p \total q}{(2\pi)^2} &\geq 0 \eqend{,} \\
\iint \tilde{W}_{vv}(p,q) \, \tilde{h}(-p) \tilde{h}^*(q) \frac{\total p \total q}{(2\pi)^2} &\geq 0 \eqend{,} \\
\iint \tilde{W}_{uv}(p,q) \Big[ \tilde{h}(-q) \tilde{g}^*(p) + \tilde{g}(-p) \tilde{h}^*(q) \Big] \frac{\total p \total q}{(2\pi)^2} &\geq 0 \eqend{,}
\end{equations}
where now $g,h \in \mathcal{S}(\mathbb{R})$ with $\tilde{g}(0) = \tilde{h}(0) = 0$.\footnote{For example, the condition for $\tilde{W}_{uu}$ is obtained by taking $\tilde{f}(p,q)$ such that $\tilde{f}(0,q) = 0$ and defining $\tilde{g}(p) \equiv \tilde{f}(p,0)$.} Since by assumption $W$ and thus its components $W_{uu}$ and so on are smooth functions, the Fourier coefficients $\tilde{W}_{uu}$ and so on have fast decay at infinity. We can thus define
\begin{splitequation}
W_M(x,y) &\equiv \iint \tilde{W}_{uu}(p,q) \, \mathe^{\mathi p u(x) + \mathi \frac{M^2}{4 p} v(x)} \mathe^{\mathi q u(y) + \mathi \frac{M^2}{4 q} v(y)} \frac{\total p \total q}{(2\pi)^2} \\
&\quad+ \iint \tilde{W}_{vu}(p,q) \, \mathe^{\mathi \frac{M^2}{4 p} u(x) + \mathi p v(x)} \mathe^{\mathi q u(y) + \mathi \frac{M^2}{4 q} v(y)} \frac{\total p \total q}{(2\pi)^2} \\
&\quad+ \iint \tilde{W}_{uv}(p,q) \, \mathe^{\mathi p u(x) + \mathi \frac{M^2}{4 p} v(x)} \mathe^{\mathi \frac{M^2}{4 q} u(y) + \mathi q v(y)} \frac{\total p \total q}{(2\pi)^2} \\
&\quad+ \iint \tilde{W}_{vv}(p,q) \, \mathe^{\mathi \frac{M^2}{4 p} u(x) + \mathi p v(x)} \mathe^{\mathi \frac{M^2}{4 q} u(y) + \mathi q v(y)} \frac{\total p \total q}{(2\pi)^2} \eqend{,}
\end{splitequation}
where the integrals are absolutely convergent and which by construction fulfills
\begin{equation}
\left( \partial_x^2 - M^2 \right) W(x,y) = \left( \partial_y^2 - M^2 \right) W(x,y) = 0 \eqend{.}
\end{equation}
Note that in the case where $W(x,y)$ is the Fourier transform of a positive measure, we have $W_{uv} = W_{vu} = 0$ and $\tilde{W}_{uu}(p,q) \total p \total q = \delta(p+q) \total \mu_{W,u}(p)$, $\tilde{W}_{vv}(p,q) \total p \total q = \delta(p+q) \total \mu_{W,v}(p)$, and the formulas above need to be adjusted accordingly. For example, this is the case for the state~\eqref{eq:thm_renorm_w_example}, where we compute
\begin{splitequation}
\tilde{W}_{uu}(p,q) \total p \total q &= \tilde{W}_{vv}(p,q) \total p \total q \\
&= \delta(p+q) \Theta\left( \abs{p} \in [E_0, E_1] \right) \frac{\pi \mathe^{- \beta \abs{p}}}{\abs{p} \left( 1 - \mathe^{- \beta \abs{p}} \right)} \total p \total q
\end{splitequation}
and
\begin{splitequation}
W_M(x,y) &= \frac{1}{4 \pi} \int \Theta\left( \abs{p} \in [E_0, E_1] \right) \frac{\mathe^{- \beta \abs{p}}}{\abs{p} \left( 1 - \mathe^{- \beta \abs{p}} \right)} \\
&\qquad\times \left[ \mathe^{\mathi p u(x,y) + \mathi \frac{M^2}{4 p} v(x,y)} + \mathe^{\mathi \frac{M^2}{4 p} u(x,y) + \mathi p v(x,y)} \right] \total p \eqend{.}
\end{splitequation}

We then define the quasi-free state $\tilde{\omega}^{\Lambda,\epsilon}$ with two-point function
\begin{equation}
\tilde{\omega}^{\Lambda,\epsilon}\left( \phi(x) \star \phi(y) \right) = \frac{1}{2 \pi} \bessel*{K}{0}{2 \Lambda \mathe^{-\gamma} \sqrt{ ( \epsilon + \mathi u(x,y) ) ( \epsilon + \mathi v(x,y) ) } } + W_{2 \Lambda \mathe^{-\gamma}}(x,y) \eqend{,}
\end{equation}
which is the two-point function in the massive theory with mass $M = 2 \Lambda \mathe^{-\gamma}$. For the smeared two-point function, it follows that
\begin{splitequation}
\label{eq:qei_omega_mass}
&\iint \tilde{\omega}^{\Lambda,\epsilon}\left( \phi(x) \star \phi(y) \right) f(x) f^*(y) \total^2 x \total^2 y \\
&= \frac{1}{2 \pi} \int_0^\infty \abs{ \int f(x) \exp\left[ - \mathi a \Lambda \mathe^{-\gamma} (x^0-x^1) - \mathi \frac{\Lambda \mathe^{-\gamma}}{a} (x^0+x^1) \right] \total^2 x }^2 \\
&\qquad\times \exp\left[ - a \Lambda \mathe^{-\gamma} \epsilon - \frac{\Lambda \mathe^{-\gamma}}{a} \epsilon \right] \frac{1}{2 a} \total a + \iint W_{2 \Lambda \mathe^{-\gamma}}(x,y) f(x) f^*(y) \total^2 x \total^2 y
\end{splitequation}
for an arbitrary test function $f \istest$, where we could interchange the integral over $a$ with the ones over $x$ and $y$ for all $\Lambda,\epsilon > 0$ because of absolute convergence, and the first term is obviously positive. To show positivity for the second term, we pass to Fourier space and obtain
\begin{splitequation}
\label{eq:qei_wmass_ft}
&\iint W_M(x,y) f(x) f^*(y) \total^2 x \total^2 y \\
&= \iint \tilde{W}_{uu}(p,q) \, \tilde{g}(-p) \tilde{g}^*(q) \frac{\total p \total q}{(2\pi)^2} + \iint \tilde{W}_{vv}(p,q) \, \tilde{h}(-p) \tilde{h}^*(q) \frac{\total p \total q}{(2\pi)^2} \\
&\quad+ \iint \tilde{W}_{uv}(p,q) \, \left[ \tilde{g}(-p) \tilde{h}^*(q) + \tilde{g}^*(p) \tilde{h}(-q) \right] \frac{\total p \total q}{(2\pi)^2} \eqend{,}
\end{splitequation}
where
\begin{equation}
\tilde{g}(p) \equiv \tilde{f}\left( p, \frac{M^2}{4 p} \right) = \iint f(u,v) \mathe^{- \mathi p u - \mathi \frac{M^2}{4 p} v} \total u \total v \eqend{,} \quad \tilde{h}(p) \equiv \tilde{f}\left( \frac{M^2}{4 p}, p \right) \eqend{.}
\end{equation}
Since the Fourier transform is an isomorphism between Schwartz spaces, it follows that $\tilde{f} \istest$, and in particular for all $M > 0$ we have
\begin{equation}
\lim_{p \to 0} \tilde{g}(p) = \lim_{p \to 0} \tilde{h}(p) = 0 \eqend{.}
\end{equation}
Using the conditions~\eqref{eq:qei_wuuvv_pos}, we see that each term in~\eqref{eq:qei_wmass_ft} is non-negative, and hence we have $\iint W_M(x,y) f(x) f^*(y) \total^2 x \total^2 y \geq 0$ for all $f \istest$, so that~\eqref{eq:qei_omega_mass}
\begin{equation}
\label{eq:qei_omega_pos}
\iint \tilde{\omega}^{\Lambda,\epsilon}\left( \phi(x) \star \phi(y) \right) f(x) f^*(y) \total^2 x \total^2 y \geq 0
\end{equation}
for all $\Lambda, \epsilon > 0$. We then obtain that
\begin{equation}
\iint f(\tau) f(\tau') v^\mu(\tau) v^\nu(\tau') \, \tilde{\omega}^{\Lambda,\epsilon}\left( \Big( \mathcal{T}_\text{int}\left[ \partial_\mu \phi(z) \right] \Big)^\dagger \star \mathcal{T}_\text{int}\left[ \partial_\nu \phi(z') \right] \right) \total \tau \total \tau' \geq 0
\end{equation}
for all $\Lambda, \epsilon > 0$, and thus also in the limit where the regulators vanish, and the same with $v^\mu v^\nu$ replaced by $w^\mu w^\nu$. Since by construction
\begin{equation}
\label{eq:qei_limit_state}
\lim_{\Lambda \to 0} \omega^{\Lambda,\epsilon}(F) = \lim_{\Lambda \to 0} \tilde{\omega}^{\Lambda,\epsilon}(F)
\end{equation}
whenever the limit is finite, it follows that the first term of~\eqref{eq:qei_e_positivebounded} is positive, and integrating over $\tau$ we obtain a positive contribution to $E_\omega(f^2)$~\eqref{eq:qei_e_fourier}.

For the remaining terms in~\eqref{eq:qei_e_positivebounded}, we want to show that they are bounded. We start with the Hadamard parametrix $H^+(z,z')$, that is the series
\begin{splitequation}
\label{eq:qei_h_sum}
&- \frac{\mathi}{2} f(\tau) f(\tau') \Big[ v^\mu(\tau) v^\nu(\tau') + w^\mu(\tau) w^\nu(\tau') \Big] \\
&\quad\times \sum_{n=0}^\infty \sum_{k=0}^n \sum_{\ell=0}^k \sum_{m=0}^{n-k} \frac{(-1)^k \mathi^n \delta_{2(\ell+m),n}}{\ell! (k-\ell)! m! (n-k-m)!} \int\dotsi\int \mathcal{E}_{k,\ell,n,m}(\vec{x};\vec{y}) \\
&\quad\times \partial^z_\mu \partial^{z'}_\nu H^+(z,z') \prod_{i=1}^{\ell+m} g(x_i) \total^2 x_i \prod_{j=1}^{n-\ell-m} g(y_j) \total^2 y_j \eqend{.}
\end{splitequation}
Since the sum is independent of $z$ and $z'$, we can simplify it further. Namely, from the Sine--Gordon interaction~\eqref{eq:sint}, the results~\eqref{eq:timeordered_exponential} and~\eqref{eq:antitimeordered_exponential} for the (anti-)time-ordered products of vertex operators, the formula~\eqref{eq:normal_ordering_starproduct} for the star product and the explicit form of the two-point function~\eqref{eq:twopf_hadamard} we obtain
\begin{splitequation}
&\overline{\mathcal{T}}\left[ \mathe_\otimes^{- \mathi S_\text{int}} \right] \star \mathcal{T}\left[ \mathe_\otimes^{\mathi S_\text{int}} \right] = \sum_{n=0}^\infty \sum_{k=0}^n \frac{(-1)^k \mathi^n}{k! (n-k)!} \int\dotsi\int \sum_{\sigma_j = \pm 1} \overline{\mathcal{T}}\left[ \bigotimes_{i=1}^k V_{\sigma_i \beta}(x_i) \right] \\
&\hspace{10em} \star \mathcal{T}\left[ \bigotimes_{j=k+1}^n V_{\sigma_j \beta}(x_j) \right] \prod_{j=1}^n g(x_j) \total^2 x_j \\
&\quad= \sum_{n=0}^\infty \sum_{k=0}^n \sum_{\ell=0}^k \sum_{m=0}^{n-k} \frac{(-1)^k \mathi^n}{\ell! (k-\ell)! m! (n-k-m)!} \int\dotsi\int \mathcal{E}_{k,\ell,n,m}(\vec{x};\vec{y}) \\
&\qquad\times \mathcal{N}_G\left[ \prod_{i=1}^{\ell+m} V_\beta(x_i) \prod_{j=1}^{n-\ell-m} V_{-\beta}(y_j) \right] \\
&\qquad\times \left( \frac{\Lambda}{\mu} \right)^\frac{\beta^2 (2(\ell+m)-n)^2}{4 \pi} \prod_{i=1}^{\ell+m} g(x_i) \total^2 x_i \prod_{j=1}^{n-\ell-m} g(y_j) \total^2 y_j \raisetag{2.3em}
\end{splitequation}
by a long but straightforward computation. As before, we used that the (anti-)time-ordered products are symmetric in their arguments, that there are $k!/[ \ell! (k-\ell)! ]$ ways to choose $\ell$ vertex operators $V_\beta$ with positive sign from a total of $k$ ones $V_{\pm \beta}$ with either sign, and renamed the insertion points of the $V_{-\beta}$ operators to $y_j$. Taking the expectation value in the state $\omega^{\Lambda,\epsilon}$ and the limit $\Lambda \to 0$, only terms with $n = 2(\ell+m)$ give a non-vanishing contribution, and renaming $n \to 2n$ we obtain
\begin{splitequation}
\label{eq:qei_1_expansion}
&\sum_{n=0}^\infty \sum_{k=0}^{2n} \sum_{\ell=\max(0,k-n)}^{\min(k,n)} \frac{(-1)^{k+n}}{\ell! (k-\ell)! (n-\ell)! (n-k+\ell)!} \int\dotsi\int \mathcal{E}_{k,\ell,2n,n-\ell}(\vec{x};\vec{y}) \\
&\quad\times \prod_{i=1}^n g(x_i) \total^2 x_i \prod_{j=1}^n g(y_j) \total^2 y_j = \omega\left( \overline{\mathcal{T}}\left[ \mathe_\otimes^{- \mathi S_\text{int}} \right] \star \mathcal{T}\left[ \mathe_\otimes^{\mathi S_\text{int}} \right] \right) = 1 \eqend{,} \raisetag{2.4em}
\end{splitequation}
where the last equality follows because $\overline{\mathcal{T}}\left[ \mathe_\otimes^{- \mathi S_\text{int}} \right] = \mathcal{T}\left[ \mathe_\otimes^{\mathi S_\text{int}} \right]^{\star (-1)}$~\eqref{eq:t_starinverse}. Since $g$ is arbitrary, we conclude that
\begin{equation}
\label{eq:qei_vanishing_sum}
\mathbb{S}_{\vec{x};\vec{y}} \sum_{k=0}^{2n} \sum_{\ell=\max(0,k-n)}^{\min(k,n)} \frac{(-1)^{k+n}}{\ell! (k-\ell)! (n-\ell)! (n-k+\ell)!} \mathcal{E}_{k,\ell,2n,n-\ell}(\vec{x};\vec{y}) = 0
\end{equation}
for all $n \geq 1$, where $\mathbb{S}_{\vec{x};\vec{y}}$ indicates symmetrisation of the $x_i$ and $y_j$ among themselves; since $\mathcal{E}_{k,k-\ell,2n,n-k+\ell}(\vec{y};\vec{x}) = \mathcal{E}_{k,\ell,2n,n-\ell}(\vec{x};\vec{y})$ as well as $k-\max(0,k-n) = \min(k,n)$ and $k-\min(k,n) = \max(0,k-n)$, the change of summation index $\ell \to k-\ell$ shows that the sum is already symmetric under the interchange of $x_i$ and $y_i$. That is, the sum in~\eqref{eq:qei_h_sum} collapses to the single term with $n = 0$, which is the free-theory contribution
\begin{splitequation}
\label{eq:qei_h}
h(\tau,\tau') &\equiv - \frac{\mathi}{2} f(\tau) f(\tau') \Big[ v^\mu(\tau) v^\nu(\tau') + w^\mu(\tau) w^\nu(\tau') \Big] \\
&\qquad\times \partial^z_\mu \partial^{z'}_\nu H^+(z,z') \Big\vert_{z = z(\tau), z' = z(\tau')} \eqend{.}
\end{splitequation}

To bound the contribution of~\eqref{eq:qei_h} to the total energy density $E_\omega(f^2)$, we compare the actual trajectory $z^\mu(\tau)$ with a stationary one $z^\mu = \delta^\mu_0 \tau$, and bound each part separately. For this, we use that since the curve $z^\mu(\tau)$ is future-directed we have~\eqref{eq:dudv_dtau}
\begin{equation}
\label{eq:eqi_du_dtau}
\partial_\tau u(z(\tau)) = \sqrt{ 1 + [ v^1(\tau) ]^2 } - v^1(\tau) > 0 \eqend{,}
\end{equation}
and therefore we can view $\tau$ as function of $u$ using the inverse function theorem. Following~\cite{flanagan1997}, we then define
\begin{splitequation}
\label{eq:qei_deltau_def}
\Delta^u_{\mu\nu}(\tau,\tau') &\equiv \partial_\mu^z \partial_\nu^{z'} \Big[ \ln\left( \epsilon + \mathi u(z,z') \right) - \ln\left[ \epsilon + \mathi \left( \tau(u(z)) - \tau(u(z')) \right) \right] \Big]_{z = z(\tau), z' = z(\tau')} \\
&= \left( \delta_\mu^0 - \delta_\mu^1 \right) \left( \delta_\nu^0 - \delta_\nu^1 \right) \frac{\partial}{\partial u(z)} \frac{\partial}{\partial u(z')} \\
&\quad\times \Big[ \ln[ \epsilon + \mathi u(z) - \mathi u(z') ] - \ln[ \epsilon + \mathi \tau(u(z)) - \mathi \tau(u(z')) ] \Big]_{z = z(\tau), z' = z(\tau')} \eqend{.} \raisetag{4em}
\end{splitequation}
Using that as $\epsilon \to 0$ we have
\begin{equation}
\ln[ \epsilon + \mathi u(z) - \mathi u(z') ] = \ln\abs{ u(z) - u(z') } + \frac{\mathi \pi}{2} \sgn[ u(z) - u(z') ]
\end{equation}
and the analogous expression with $\tau(u)$, and that $\sgn[ \tau(u(z)) - \tau(u(z')) ] = \sgn[ u(z) - u(z') ]$ since $\tau$ is a monotone function of $u$, we can simplify~\eqref{eq:qei_deltau_def} to
\begin{splitequation}
\Delta^u_{\mu\nu}(\tau,\tau') &= \left( \delta_\mu^0 - \delta_\mu^1 \right) \left( \delta_\nu^0 - \delta_\nu^1 \right) \frac{\partial}{\partial u(z)} \frac{\partial}{\partial u(z')} \\
&\qquad\times \Big[ \ln\abs{ u(z) - u(z') } - \ln\abs{ \tau(u(z)) - \tau(u(z')) } \Big]_{z = z(\tau), z' = z(\tau')} \\
&= \left( \delta_\mu^0 - \delta_\mu^1 \right) \left( \delta_\nu^0 - \delta_\nu^1 \right) \\
&\qquad\times \bigg[ \frac{1}{\left[ u(z) - u(z') \right]^2} - \frac{\tau'(u(z)) \tau'(u(z'))}{\left[ \tau(u(z)) - \tau(u(z')) \right]^2} \bigg]_{z = z(\tau), z' = z(\tau')}
\end{splitequation}
and
\begin{splitequation}
\lim_{\tau' \to \tau} \Delta^u_{\mu\nu}(\tau,\tau') &= \left( \delta_\mu^0 - \delta_\mu^1 \right) \left( \delta_\nu^0 - \delta_\nu^1 \right) \frac{3 [ \tau''(u(z)) ]^2 - 2 \tau'(u(z)) \tau'''(u(z))}{12 [ \tau'(u(z)) ]^2} \eqend{,}
\end{splitequation}
which was already determined in~\cite[Eq.~(2.17)]{flanagan1997}. From~\eqref{eq:eqi_du_dtau}, we obtain
\begin{equation}
\tau'(u(z)) = \left[ \partial_\tau u(z(\tau)) \right]^{-1} = \frac{1}{v^0(z(\tau)) - v^1(z(\tau))} \eqend{,}
\end{equation}
and it follows that [with $v^\mu(\tau) \equiv v^\mu(z(\tau))$]
\begin{splitequation}
\label{eq:qei_deltau}
\Delta^u_{\mu\nu}(\tau,\tau) = - \frac{\left( \delta_\mu^0 - \delta_\mu^1 \right) \left( \delta_\nu^0 - \delta_\nu^1 \right)}{3 [ v^0(\tau) - v^1(\tau) ]^\frac{3}{2}} \frac{\partial^2}{\partial \tau^2} \frac{1}{\sqrt{ v^0(\tau) - v^1(\tau) }} \eqend{.}
\end{splitequation}
The analogous computation establishes that
\begin{splitequation}
\label{eq:qei_deltav}
\Delta^v_{\mu\nu}(\tau,\tau') &\equiv \partial_\mu^z \partial_\nu^{z'} \Big[ \ln\left( \epsilon + \mathi v(z,z') \right) - \ln\left[ \epsilon + \mathi \left( \tau(v(z)) - \tau(v(z')) \right) \right] \Big]_{z = z(\tau), z' = z(\tau')} \\
&\to - \frac{\left( \delta_\mu^0 + \delta_\mu^1 \right) \left( \delta_\nu^0 + \delta_\nu^1 \right)}{3 [ v^0(\tau) + v^1(\tau) ]^\frac{3}{2}} \frac{\partial^2}{\partial \tau^2} \frac{1}{\sqrt{ v^0(\tau) + v^1(\tau) }} \quad (\tau' \to \tau) \eqend{,}
\end{splitequation}
using that also
\begin{equation}
\label{eq:eqi_dv_dtau}
\partial_\tau v(z(\tau)) = \partial_\tau \left[ z^0(\tau) + z^1(\tau) \right] = v^0(\tau) + v^1(\tau) = \sqrt{ 1 + [ v^1(\tau) ]^2 } + v^1(\tau) > 0 \eqend{.}
\end{equation}
By construction, the $\Delta^{u/v}$ measure the difference between the parametrix on the actual trajectory $z^\mu(\tau)$ and an auxiliary stationary one $z^\mu = \delta^\mu_0 \tau$, and we see that this is finite, but possibly large if the trajectory is almost light-like with $v^1 \approx \pm v^0$. In particular, it diverges in the light-like limit $v^1 \to \pm v^0$, and we see again that we cannot simply take this limit to arrive at the result for a null trajectory.

Using the explicit expression for the Hadamard parametrix~\eqref{eq:hadamard}, we can thus decompose~\eqref{eq:qei_h} in the form
\begin{splitequation}
\label{eq:qei_htau_decomp}
h(\tau,\tau') &= \frac{1}{8 \pi} f(\tau) f(\tau') \Big[ v^\mu(\tau) v^\nu(\tau') + w^\mu(\tau) w^\nu(\tau') \Big] \\
&\quad\times \bigg[ \Delta^u_{\mu\nu}(\tau,\tau') + \Delta^v_{\mu\nu}(\tau,\tau') + \partial^z_\mu \partial^{z'}_\nu \ln\left[ \epsilon + \mathi \left( \tau(u(z)) - \tau(u(z')) \right) \right] \\
&\qquad+ \partial^z_\mu \partial^{z'}_\nu \ln\left[ \epsilon + \mathi \left( \tau(v(z)) - \tau(v(z')) \right) \right] \bigg]_{z = z(\tau), z' = z(\tau')} \\
&\equiv h^\Delta(\tau,\tau') + h^0(\tau,\tau') \eqend{,} \raisetag{1.2em}
\end{splitequation}
where $h^\Delta$ is the part containing $\Delta^{u/v}$. In this part, we can take the limit $\tau' \to \tau$, which results in
\begin{splitequation}
\label{eq:qei_hdeltatau}
h^\Delta(\tau,\tau) &= - \frac{1}{12 \pi} f^2(\tau) \bigg[ \sqrt{ v^0(\tau) - v^1(\tau) } \frac{\partial^2}{\partial \tau^2} \frac{1}{\sqrt{ v^0(\tau) - v^1(\tau) }} \\
&\qquad\qquad+ \sqrt{ v^0(\tau) + v^1(\tau) } \frac{\partial^2}{\partial \tau^2} \frac{1}{\sqrt{ v^0(\tau) + v^1(\tau) }} \bigg] \\
&= - \frac{1}{24 \pi} f^2(\tau) \frac{[ \dot v^1(\tau) ]^2}{1 + [ v^1(\tau) ]^2} \eqend{,}
\end{splitequation}
where we used that $w^1 = v^0$, $w^0 = v^1$ and that $v^0(\tau) = \sqrt{1 + [ v^1(\tau) ]^2}$, which follows from the normalisation $v^\mu v_\mu = -1$. For $h^0(\tau,\tau')$ which contains the contribution from the straight trajectory, we replace $f$ by a real-valued test function $f_N$ with compact support, call the resulting expression $h^0_N$ and bound the contribution of this term to the energy density using~\eqref{eq:thm_qei_N}. We thus have to compute
\begin{equation}
\label{eq:qei_htautau_fourier}
\int h^0_N(\tau,\tau) \total \tau = 2 \lim_{\delta \to 0^+} \int_0^\infty \mathe^{- \delta \xi^2} (\mathcal{F} h^0_N)(-\xi,\xi) \frac{\total \xi}{2 \pi}
\end{equation}
with the Fourier transform~\eqref{eq:qei_fourier_def}
\begin{splitequation}
\label{eq:qei_fourier_hn}
(\mathcal{F} h^0_N)(\xi,\zeta) &= \iint \mathe^{\mathi \xi \tau + \mathi \zeta \sigma} h^0_N(\tau,\sigma) \total \tau \total \sigma \\
&= \frac{1}{2 \pi} \iint \mathe^{\mathi \xi \tau + \mathi \zeta \sigma} f_N(\tau) f_N(\sigma) \frac{1}{( \tau - \sigma - \mathi \epsilon )^2} \total \tau \total \sigma \eqend{,}
\end{splitequation}
where we performed the derivatives in $h^0_N$~\eqref{eq:qei_htau_decomp} using~\eqref{eq:eqi_du_dtau} and~\eqref{eq:eqi_dv_dtau} and used again that $w^1 = v^0$, $w^0 = v^1$.

Using the Fourier transform
\begin{equation}
\frac{1}{( \tau - \mathi \epsilon )^2} = \int \left[ - 2 \pi p \, \Theta(p) \mathe^{- p \epsilon}  \right] \mathe^{- \mathi p \tau} \frac{\total p}{2 \pi} \eqend{,}
\end{equation}
we can write~\eqref{eq:qei_fourier_hn} as
\begin{equation}
\label{eq:qei_fourier_hn2}
(\mathcal{F} h^0_N)(\xi,\zeta) = - \int \tilde{f}_N(p-\xi) \tilde{f}_N(-p-\zeta) p \, \Theta(p) \mathe^{- p \epsilon} \frac{\total p}{2 \pi} \eqend{,}
\end{equation}
where
\begin{equation}
\tilde{f}_N(q) \equiv \int f_N(\tau) \mathe^{- \mathi q \tau} \total \tau
\end{equation}
is the Fourier transform of the test function, which is also a Schwartz function, and we could interchange the integrals for $\epsilon > 0$ because everything converges absolutely. It follows that~\eqref{eq:qei_htautau_fourier}
\begin{splitequation}
\int h^0_N(\tau,\tau) \total \tau &= - 2 \lim_{\delta \to 0^+} \int_0^\infty \mathe^{- \delta \xi^2} \int \tilde{f}_N(p+\xi) \tilde{f}_N(-p-\xi) p \, \Theta(p) \mathe^{- p \epsilon} \frac{\total p}{2 \pi} \frac{\total \xi}{2 \pi} \\
&= - 2 \int_0^\infty \int_0^\infty \abs{ \tilde{f}_N(p+\xi) }^2 p \, \mathe^{- p \epsilon} \frac{\total p}{2 \pi} \frac{\total \xi}{2 \pi} \eqend{,} \raisetag{1.8em}
\end{splitequation}
where we used that since $f_N \in \mathcal{S}(\mathbb{R})$ is real-valued, we have $\tilde{f}_N(-p) = \tilde{f}^*_N(p)$, and where we could take the limit inside the integral because everything converges absolutely. Here we see that since $p$ is positive, for a non-trivial quantum energy inequality we need that $\xi \geq 0$~\cite{fewstersmith2008}, since otherwise the integrand would be unbounded along the diagonal $\xi = -p$ in the limit $\epsilon \to 0$. Shifting $p \to p - \xi$, we can then interchange the integrations and perform the integral over $\xi$, which gives
\begin{splitequation}
\label{eq:qei_h0n_tauint}
\int h^0_N(\tau,\tau) \total \tau &= - 2 \int_0^\infty \abs{ \tilde{f}_N(p) }^2 \int_0^p (p-\xi) \, \mathe^{- (p-\xi) \epsilon} \frac{\total \xi}{2 \pi} \frac{\total p}{2 \pi} \\
&= - \frac{1}{2 \pi^2} \int_0^\infty \abs{ \tilde{f}_N(p) }^2 \frac{1 - \mathe^{- p \epsilon} ( 1 + p \epsilon )}{\epsilon^2} \total p \eqend{.}
\end{splitequation}
Using dominated convergence, we can then take the limit $\epsilon \to 0$ inside the integral, since the integrand (apart from the test functions) is bounded by its value at $\epsilon = 0$ (which is $p^2/2$). Using finally Parseval's theorem for the Fourier transform, we obtain
\begin{equation}
\label{eq:qei_h0n_tauint2}
\int h^0_N(\tau,\tau) \total \tau = - \frac{1}{4 \pi^2} \int_0^\infty \abs{ p \tilde{f}_N(p) }^2 \total p = - \frac{1}{4 \pi} \int \big[ f'_N(\tau) \big]^2 \total \tau \eqend{.}
\end{equation}
Choosing now a sequence of test functions $f_N(\tau) \equiv \chi_N(\tau) f(\tau)$ with $\supp \chi_N \subset [-N-2,N+2]$, $\chi_N(\tau) = 1$ for $\tau \in [-N,N]$ and $\abs{\chi_N(\tau)}, \abs{\chi'_N(\tau)} \leq 1$, we have $\abs{ f'_N(\tau) } \leq \abs{ f(\tau) } + \abs{ f'(\tau) }$ and can use dominated convergence to take the limit $N \to \infty$ inside the integral. Taking all together, from~\eqref{eq:qei_h}, \eqref{eq:qei_htau_decomp}, \eqref{eq:qei_hdeltatau} and~\eqref{eq:qei_h0n_tauint2} we thus obtain
\begin{splitequation}
\label{eq:qei_h_result}
&\int \bigg[ - \frac{\mathi}{2} f(\tau) f(\tau) \Big[ v^\mu(\tau) v^\nu(\tau) + w^\mu(\tau) w^\nu(\tau) \Big] \partial^z_\mu \partial^{z'}_\nu H^+(z,z') \Big\vert_{z = z' = z(\tau)} \bigg] \total \tau \\
&\quad= - \frac{1}{24 \pi} \int \left[ 6 \big[ f'(\tau) \big]^2 + f^2(\tau) \frac{[ \dot v^1(\tau) ]^2}{1 + [ v^1(\tau) ]^2} \right] \total \tau \eqend{,} \raisetag{1.8em}
\end{splitequation}
which depends on the test function $f$ and the trajectory $z^\mu(\tau)$ and is clearly negative, but bounded.

It remains to bound the remaining terms in~\eqref{eq:qei_e_positivebounded}, which contain $\mathcal{G}_{k,\ell,n,m}$ and the vertex operators $V_{\pm \beta}$. For the first kind of terms, we use the definition of $\mathcal{G}_{k,\ell,n,m}$~\eqref{eq:calg_def} and $\overline{\mathcal{G}}_{k,\ell,n,m}$~\eqref{eq:barcalg_def} as well as the definitions of the Feynman~\eqref{eq:hf_def} and Dyson parametrix~\eqref{eq:hd_def} to obtain
\begin{splitequation}
&\mathcal{G}_{k,\ell,2(\ell+m),m}(\vec{x};\vec{y};z) - \overline{\mathcal{G}}_{k,\ell,2(\ell+m),m}(\vec{x};\vec{y};z) \\
&= \sum_{i=1}^\ell \Big[ H^+(x_i,z) - H^\text{D}(z,x_i) \Big] + \sum_{i=\ell+1}^{\ell+m} \Big[ H^\text{F}(x_i,z) - H^+(z,x_i) \Big] \\
&\quad\qquad- \sum_{j=1}^{k-\ell} \Big[ H^+(y_j,z) - H^\text{D}(z,y_j) \Big] - \sum_{j=k-\ell+1}^{\ell+m} \Big[ H^\text{F}(y_j,z) - H^+(z,y_j) \Big] \\
&= \sum_{i=1}^{\ell+m} \Big[ G_\text{ret}(x_i,z) - G_\text{ret}(y_i,z) \Big] \raisetag{2.2em}
\end{splitequation}
with the state-independent retarded propagator
\begin{equation}
G_\text{ret}(x,y) \equiv \Theta(x^0-y^0) \Big[ H^+(x,y) - H^+(y,x) \Big] \eqend{.}
\end{equation}
Replacing $n \to 2n$ and using the Kronecker $\delta$ for the sum over $m$, it thus follows that the contribution of terms containing $\mathcal{G}_{k,\ell,n,m}$ to~\eqref{eq:qei_e_positivebounded} can be written as
\begin{splitequation}
&\frac{\beta^2}{2} f(\tau) f(\tau') \Big[ v^\mu(\tau) v^\nu(\tau') + w^\mu(\tau) w^\nu(\tau') \Big] \\
&\quad\times \sum_{n=0}^\infty \sum_{k=0}^{2n} \sum_{\ell=\max(0,k-n)}^{\min(k,n)} \frac{(-1)^{k+n}}{\ell! (k-\ell)! (n-\ell)! (n-k+\ell)!} \int\dotsi\int \mathcal{E}_{k,\ell,2n,n-\ell}(\vec{x};\vec{y}) \\
&\quad\times \sum_{i=1}^n \partial^z_\mu \Big[ G_\text{ret}(x_i,z) - G_\text{ret}(y_i,z) \Big] \partial^{z'}_\nu \bigg[ \sum_{i=1}^\ell H^+(x_i,z') + \sum_{i=\ell+1}^n H^\text{F}(x_i,z') \\
&\quad\qquad- \sum_{j=1}^{k-\ell} H^+(y_j,z') - \sum_{j=k-\ell+1}^n H^\text{F}(y_j,z') - \mathi \sum_{i=1}^n W(x_i,z') + \mathi \sum_{j=1}^n W(y_j,z') \bigg] \\
&\quad\times \prod_{i=1}^n g(x_i) g(y_i) \total^2 x_i \total^2 y_i \eqend{,} \raisetag{2.2em}
\end{splitequation}
which contains state-dependent terms involving $W$ and state-independent ones. Since the state-dependent terms are independent of $k$ and $\ell$, and moreover symmetric under the interchange of the $x_i$ and $y_j$ among themselves, we can use again the identity~\eqref{eq:qei_vanishing_sum} to conclude that they vanish for all $n > 0$. However, since they only appear starting from $n = 1$, they do not contribute at all, and only state-independent terms remain. For those, we can take the limit $\tau' \to \tau$ and bound them as in the previous section~\ref{sec_renorm_qed}, and the bound is state-independent due to our condition~\eqref{eq:thm_renorm_wpos_disc} on $W$. In the same way, we bound the contribution of the vertex operators in~\eqref{eq:qei_e_positivebounded}, and the resulting bound is also state-independent.

Taking all together, we have shown that the first term in~\eqref{eq:qei_e_positivebounded} is state-dependent but positive, since it can be obtained as the limit of vanishing cutoff $\Lambda \to 0$~\eqref{eq:qei_limit_state} of an alternative quasi-free state with positive definite two-point function~\eqref{eq:qei_omega_pos}. For the terms involving the Hadamard parametrix $H^+(z,z')$, the only non-vanishing contribution comes from the free theory~\eqref{eq:qei_h}, which we have split into a part evaluated on a straight worldline and a part describing the difference between this and the actual worldline. The second part has a finite limit as $\tau' \to \tau$~\eqref{eq:qei_hdeltatau}, while for the first part we have used~\cite[Lemma~2.12]{fewstersmith2008} in the form~\eqref{eq:thm_qei_N}; taking both together we obtain a negative but state-independent contribution~\eqref{eq:qei_h_result} to the energy density. The remaining terms in~\eqref{eq:qei_e_positivebounded} can be simply bounded, noticing that the derivatives of the state-dependent part $W$ drop out. Since the only state dependence is then in the exponentials which we estimate by $1$ using our condition~\eqref{eq:thm_renorm_wpos_disc} on $W$, the bound on these terms is state-independent. The sum of the bounds for these terms together with the contribution~\eqref{eq:qei_h_result} from the Hadamard parametrix then gives a lower bound on the smeared energy density~\eqref{eq:thm_qei}, which depends on the worldline $z(\tau)$, the smearing $f$, $\beta$ and the adiabatic cutoff $g$. Furthermore, this bound is non-trivial since the state-dependent positive part is arbitrary and can thus certainly be larger. \hfill\squareforqed

\begin{acknowledgements}
This work is supported by the Deutsche Forschungsgemeinschaft (DFG, German Research Foundation) --- project no. 396692871 within the Emmy Noether grant CA1850/1-1 and project no. 406116891 within the Research Training Group RTG 2522/1. It is a pleasure to thank Henning Bostelmann, Chris Fewster, Ben Freivogel, Nicola Pinamonti and Aron Wall for comments and suggestions. We also thank Fabrizio Zanello for discussions on conserved currents during the 46th LQP workshop in Erlangen.
\end{acknowledgements}


\bibliography{literature}

\end{document}